\documentclass[a4paper]{tlp}

\usepackage{latexsym}
\usepackage{graphicx}
\usepackage{aopmath}
\usepackage{rotating}
\usepackage{float} 
\usepackage{subfig}
\usepackage{amssymb}

\title{The role of semantics \\ in mining frequent patterns from knowledge bases \\ in description logics with rules}
\shorttitle{The role of semantics in mining frequent patterns }

\author[J. J\'ozefowska \and A. \L{}awrynowicz \and T. \L{}ukaszewski]
{
JOANNA J\'OZEFOWSKA, AGNIESZKA \L{}AWRYNOWICZ, TOMASZ \L{}UKASZEWSKI\\
\ \ \ \ \ \ \ Institute of Computing Science, Poznan University of Technology,\ \ \ \ \ \ \ \ \\ 
\ \ \ \ \ \ \ \ \ ul. Piotrowo 2, 60-965 Poznan, Poland\ \ \ \ \ \ \ \ \ \\
\ \ \ \ \ \ \ Email: \{jjozefowska, alawrynowicz, tlukaszewski\}@cs.put.poznan.pl\ \ \ \ \ \ \ 
}

\date{}

\begin{document}
\bibliographystyle{acmtrans}

\long\def\comment#1{}

\maketitle

\begin{abstract}
We propose a new method for mining frequent patterns in a language that combines both Semantic Web ontologies and rules. In particular we consider the setting of using a language that combines description logics with DL-safe rules. This setting is important for the practical application of data mining to the Semantic Web. We focus on the relation of the semantics of the representation formalism to the task of frequent pattern discovery, and for the core of our method, we propose an algorithm that exploits the semantics of the combined knowledge base.
We have developed a proof-of-concept data mining implementation of this. 
Using this we have empirically shown that using the combined knowledge base to perform semantic tests can make data mining faster by pruning useless candidate patterns before their evaluation. We have also shown that the quality of the set of patterns produced may be improved: the patterns are more compact, and there are fewer patterns. We conclude that exploiting the semantics of a chosen representation formalism is key to the design and application of (onto-)relational frequent pattern discovery methods.  
\\
\\
\textbf{Note: To appear in Theory and Practice of Logic Programming (TPLP)}.
\end{abstract}

\begin{keywords}
frequent pattern discovery, ontologies, Semantic Web, DL-safe rules
\end{keywords}

\section{Introduction}
The discovery of frequent patterns is a fundamental data mining task. It has been studied for many different forms of input data and the pattern.  Within the  relational setting it has been investigated 
since the development of \textsc{WARMR} \cite{WARMR_99}. 
\textsc{WARMR} uses the Datalog subset of \emph{first}-\emph{order logic (FOL)} as the representation language for both data and patterns. As such, \textsc{WARMR}, and other subsequently proposed relational frequent pattern miners, \textsc{FARMER} \cite{FARMER:2001,FARMER:2003} and c-armr \cite{carmr}, can be classified as \emph{Inductive Logic Programming (ILP)} \cite{NC_Wolf:1997,RDM_book} methods. These ILP systems have been successfully applied to a number of domains, most notably bioinformatics \cite{KingYeast2000,KingKCD00,KingBioinformatics2001} and chemoinformatics \cite{Dehaspe98findingfrequent}. 
\\
\indent
While relational frequent pattern mining methods have mostly assumed Datalog as the representation language, currently most activity within the field of \emph{knowledge representation (KR)} assumes the use of logic-based ontology languages such as \emph{description logics (DLs)} \cite{DL_handbook}. 
Thanks to its significant support for modelling ontologies, and suitability to the inherently open and incomplete nature of the Web environment, description logic has been chosen as the formal foundation of the standard ontology language for the Web, the \emph{Web Ontology Language (OWL)} \cite{OWL}. OWL is now considered one of the fundamental technologies underpinning  the \emph{Semantic Web} \cite{tbl:2001}, currently one of the most active application fields of artificial intelligence. 
\\
\indent
Research in KR is focused on developing deductive reasoning procedures, which are also traditionally employed to reason with logic-based ontological data. However, to meet the challenges posed by the Semantic Web scale and use cases, such deductive approaches are not enough. Therefore, there is a recent trend in Semantic Web research to propose complementary forms of reasoning that are more efficient and noise-tolerant. A promising approach in this area is to use inductive methods to complement deductive ones.  This is in line with the recent trends in ILP research to broaden the scope of the 
logical formalisms considered to description logics, or hybrid languages combining description logics with logic programs. Since description logic knowledge bases are often equated with ontologies, ILP methods applied to such knowledge bases have been referred to as ''ontology mining'' methods \cite{fanizzi06declarative,1483172,DBLP:conf/ekaw/dAmatoSF08}, and the ones applied to the hybrid knowledge bases to as ''onto-relational mining'' methods \cite{DBLP:conf/ilp/LisiE08}. 
To the best of our knowledge, only one onto-relational frequent pattern mining method, \textsc{SPADA} \cite{Lisi:2004}, has been proposed. 
\\
\indent
This paper describes a method for frequent pattern mining in knowledge bases represented in the formalism of DL-safe rules \cite{dl_safe_rules2} that combine Semantic Web ontologies (represented in description logic) and rules (represented in disjunctive Datalog). 
This language meets the requirements of knowledge representation for the Semantic Web and target application domains, and possesses properties suitable for data mining applications. 
In the core of the method, we propose an algorithm that exploits at various steps the semantics of a combined knowledge base on which it operates. We show how to realize the proposed method in terms of exploiting state-of-the-art reasoning techniques, and present a proof-of-concept implementation of the method. 
For Semantic Web research, the paper contributes to the general understanding of the role of ontologies and semantics in helping to solve knowledge-intensive tasks by exploiting the meaning of the represented knowledge. 
For ILP data mining research, the method's main novel feature is its exploitation of the semantics of the chosen language.
\\
\indent
The rest of the paper is organized as follows. Section \ref{sec:motivation} disusses a technical and an application-oriented motivation of the work. In Section \ref{sec:prelim} we introduce the basics of knowledge representation formalisms considered in this paper, and the problem of frequent pattern mining from combined knowledge bases. In Section \ref{sec:algorithm} we present our method for mining frequent patterns. In Section \ref{sec:experiments} we present the experimental evaluation of the proposed approach. Section \ref{sec:related} contains the discussion of the related work. Finally, Section \ref{sec:conc} concludes the paper, and outlines future work.

\section{Motivation}\label{sec:motivation}
\subsection{The Setting}
The problem of combining ontologies with rules is central in the Semantic Web. In the current stack of the Semantic Web languages, rules are placed in the same layer as ontologies. There is an ongoing initiative to define an open format for rule interchange on the Semantic Web, the \emph{Rule Interchange Format (RIF)}\footnote{http://www.w3.org/2005/rules/wiki/RIF\_Working\_Group}, that will cover a wide spectrum of rule types, among them deductive rules represented in Datalog. As we will discuss further in the paper, some important application domains such as life sciences require a language that combines description logic with some form of Datalog rules. 
\\
\indent
Since a straightforward combination of DL and rules may easily lead to the undecidability of reasoning problems, the problem of developing such combinations has received a lot of attention in KR and Semantic Web research.   
This has resulted in several proposals which may be generally divided into the following approaches: interaction of rules and ontologies with strict semantic separation (loose coupling), interaction of rules and ontologies with strict semantic integration (tight coupling), and reductions from DLs to logic programming  formalisms. 
\\
\indent
In the first approach, adopted by \emph{dl}-\emph{programs} \cite{dl-programs1,dl-programs2,DBLP:journals/ai/EiterILST08}, DL and rule components are technically separate, and can be seen as black boxes communicating via ''safe interface''.  
\\
\indent
In the second type of approach, ''safe interaction'', rules and DL knowledge bases are combined in a common semantic framework. A straightforward, tight extension of DL with first-order implication as proposed for \emph{Semantic Web Rule Language (SWRL)} in \cite{SWRL}, is trivially undecidable. On the other hand, \emph{Description Logic Programs (DLP)} \cite{DLP_Grosof} describe a decidable intersection of description logic and logic programs. In between of these two opposite approaches, there is a group of proposals such as $\mathcal{AL}$-\emph{log} \cite{AL_log}, \emph{CARIN} \cite{CARIN}, \emph{DL}-\emph{safe rules} \cite{dl_safe_rules2} or $\mathcal{DL}$+\emph{log} \cite{dl_plus_log} where to obtain decidability, either DL, or rules or both are typically constrained by various syntactic restrictions, e.g. in the form of a safety condition.
However, the usual syntactic restrictions may also be dropped, through changing the usual perspective of the integration from DLs to the perspective of rule-based systems, as proposed in \cite{DBLP:conf/esws/Lukasiewicz07} for the case of a tightly integrated form of disjunctive dl-programs. 
Finally, the tight integration may also become a full one, as in \emph{hybrid MKNF knowledge bases} \cite{motik_faithful}, where there is no separation between vocabularies of a DL and rule component.  
\\
\indent
An interesting representative approach for the works consisting on reducing description logics to logic programming is proposed in \cite{KAON2_REDUCTION:2004,KAON2_REDUCTION:2007} for an expressive DL language $\mathcal{SHIQ}$. 
In that approach, the consistency checking and query answering is reduced to the evaluation of a positive disjunctive Datalog program, which is obtained by a translation of a description logic knowledge base to first-order logic, followed by an application of superposition techniques, and the elimination of function symbols from the resulting set of clauses. 
\\
\indent
Despite of the desired expressivity, there are also other important requirements for a language to be used in frequent pattern mining applications, which are data-intensive in nature. In the last decade, the focus of KR research has been mostly on developing reasoning techniques for handling complex DL intensional knowledge, and on decidability issues. However, new Semantic Web applications require efficient scalable procedures for query answering over ontologies, which now becomes an intensively explored area of research. Scalability may be achieved by restricting features of a DL language to obtain a lightweight one, but tailored for data-intensive applications, as in the case of a tractable family of   lanuages called \emph{DL}-\emph{Lite} \cite{1286097}. An interesting recent study into this direction is presented in \cite{DBLP:conf/pods/CaliGL09}, where a family of expressive extensions of Datalog is proposed that generalize the DL-Lite family, e.g. by admitting  existentially quantified variables in rule heads. 
The requirement for efficient query answering over large amounts of data (extensional knowledge) is crucial for frequent pattern mining applications.
\\
\indent
Taking into account both criteria, that is sufficiently interesting expressivity required for real applications, and efficient query answering procedures, one combination of DLs and rules with interesting properties is the formalism of DL-safe rules \cite{dl_safe_rules2}. In this formalism decidability is obtained by restricting the rules to DL-safe ones that are applicable only to instances explicitly known by name.   
As it was shown in \cite{KAON2_REDUCTION:2004,KAON2_REDUCTION:2007}, the restriction to DL-safety enables the transformation of a DL knowledge base to a disjunctive Datalog program. This in turn enables the application of well-known reasoning algorithms and optimization techniques (such as magic-sets or join-order optimizations) developed for deductive databases in order to handle large data quantities. Some of these methods have recently been extended for disjunctive Datalog \cite{magic_set_disj:2004}. The algorithm proposed in \cite{dl_safe_rules2} for query answering in DL with DL-safe rules separates reasoning on the intensional part of a knowledge base from that on the extensional part, which means that the inferences made on the intensional part are not repeated for different instances during query answering. This in turn enables better complexity results for the query answering algorithm than in case of the other state-of-the-art reasoning techniques developed for expressive DLs \cite{KAON2_COMPLEXITY:2005,motik_aboxes}. It should be noted that, if the translation does not generate any disjunctive rules, then the algorithm applies the least fixpoint operator used to evaluate non-disjunctive Datalog programs. Since the consequences of the least fixpoint operator can be computed in polynomial time, an important feature of the algorithm is that its behaviour becomes tractable while it is applied for less expressive languages.  
\\
\indent
The discussed features make the chosen DL-safe rules formalism suitable for the envisaged frequent pattern mining applications.  

\subsection{Possible applications}
The primary motivation for our work is for the application of our method 
to real-world data-mining applications.  
Arguably the most extensive use of Semantic Web KR methods is in 
the domain of biology.  Large amounts of data are increasingly 
becoming openly available and described using real-life ontologies, represented in Semantic Web languages, such as GO (Gene Ontology)\footnote{http://www.geneontology.org} or BioPax (biological pathway knowledge)\footnote{http://www.biopax.org}.  This opens up the possibility for interesting large-scale and real-world onto-relational data mining applications. 
\\
\indent
Below we will describe why KR in biology requires a language able to model the existence of unknown entities, disjunctions, and arbitrary composition of relations, that is a language that combines description logic with some form of Datalog rules.  These requirements are a domain specific motivation for our language selection.
\\
\indent
Information stored in biological knowledge bases is inherently incomplete.
For example, in functional genomics ''\emph{every protein has a function}'', but often this function is unknown. Similarly, it is known that certain genes exist because they encode known proteins, but the identity of these genes are unknown (so-called ''locally orphan'' genes). The existence of entities with unknown identity can be easily represented in description logic\footnote{$Protein \sqsubseteq \exists hasFunction$}, while it cannot be represented in Datalog.
\\
\indent
Another way of modelling incompletness is by use of disjunction, what is not expressible in Datalog. 
For example, disjunction may be used to describe that one instance of certain tertiary structure units must be present in a protein (''\emph{a classical tyrosine phosphatase has at least one low molecular weight phosphotyrosine \emph{or} one tyrosine specific with dual specificity p-domain}''\cite{DBLP:journals/ijmms/StevensAWSDHR07}).  
\\
\indent
In general, DLs employ the \emph{open-world assumption (OWA)} which seems suitable for a domain characterized by information that is incomplete either due to the limits in the current state of knowledge or due to omissions common in curation processes. OWA is closely related to the \emph{monotonic} form of reasoning, classically assumed in FOL. 
The monotonicity of reasoning in DLs is in line with the need for the knowledge held in scientific knowledge bases to only comprise information that is generally accepted and experimentally validated, and, which is reasonable to assume, will not be falsified.  
For example, one may state that ''\emph{In E.
coli K-12, the protein encoded by the gene ECK0647 when in inner membrane,
facilitates the transport of glutamate from the periplasm to the cytoplasm.}''\footnote{$inMembrane \sqcap inCytoplasm \sqsubseteq \bot, ECK0647\_Protein \sqsubseteq inCytoplasm, glutamateTransport \sqsubseteq \exists participant$.$ECK0647\_Protein \sqcap inMembrane$} \cite{biopax}. This statement does not provide the reference to a particular enzyme and any information whether the transport is active or passive. If we subsequently learn such information, this does not change any positive or negative conclusions. Since it is characteristic of the current state of biology that much information is only known to a certain degree, the monotonicity of reasoning in DLs 
allows scientific knowledge bases to be extensible and evolve with scientific knowledge. 
\\
\indent
However, while possessing features which are not available in Datalog, description logic also has limitations. It does not allow relations of arbitrary arity and arbitrary composition of relations. Assume for example, that we would like to express that ''\emph{whenever a metal ion is bound
to a phosphatase which catalyses a dephosphorylation of some protein, then
this ion regulates the dephosphorylation of this protein}''\cite{DBLP:journals/ijmms/StevensAWSDHR07}\footnote{$regulates(x,y)\leftarrow metal\_ion(x), isBound(x,z), phosphatase(z), catalyses(z,y), dephosphorylation(y)$}.
It requires modelling that a composition of relations implies another relation, which may be expressed in the form of a rule. 
\\
\indent
The above arguments show the need for languages combining description logic with some form of rules in the discussed field, and the DL-safe rules formalism meets all the necessary requirements for expressivity. It is also interesting to note here, that in \cite{biopax}, while discussing problems concerned with using description logic for modelling metabolic pathways, it has been already argued, there should be a way for modelling that certain axioms may be applied only if there are known instances of some class. The solution to check this could be by submitting a query with a DL-safety condition.
\\
\indent
An interesting sample application of frequent pattern mining in the field of biology may be in functional genomics with the goal of the identification of frequent patterns in the amino acid sequence descriptions. The results would be further used to generate rules for predicting protein functional class. Such an approach would constitute an onto-relational upgrade of the relational data mining application already proposed in the literature \cite{KingKCD00}.
Another, novel application may be in metabolic pathways analysis. The goal of the application would be to identify common frequent pathways in human and other organisms that cause human diseases. The results of such analysis would allow for targeted drug design. 
\\
\indent
Despite of life sciences, frequent pattern mining applications on the Semantic Web data may be valuable in many other domains. Let us take e-business as another example. As rules provide a powerful business logic representation many use cases\footnote{http://www.w3.org/TR/rif-ucr} provided for RIF are actually in this domain.  
Most value for e-business that combinations of description logic ontologies and rules may provide is in increasing interoperability. Description logic provide means for expressing common vocabularies and domain knowledge, while rules enable to explicitely express business policies.   
For example, annotation of product and service offerings with terms from common ontologies such as GoodRelations\footnote{http://www.heppnetz.de/projects/goodrelations} may enable customers and enterprises an automatic search for suitable suppliers across the Web. 
Further, employing rules may enable to automatically express business relations between offerings and customers, and to express business policies such as ''\emph{The discount for a customer who buys a product is 5 percent
if the customer is premium and the product is regular.}''\footnote{$discount(x, y, percent5)\leftarrow premium(x), regular(y), buys(x,y)$}.
\\
\indent
In e-business domain, an interesting sample application of frequent pattern mining may be in finding frequent customer buying behaviors to support personalisation, recommendation services and targeted marketing. 
\\
\indent 
It should be stressed that for e-business applications, relatively lightweight ontologies may be sufficient, but the need for combining them with rules is essential in this domain.

\section{Preliminaries}
\label{sec:prelim}
\subsection{Language of knowledge representation}

In this section we introduce the language of knowledge representation based on the formalism of DL-safe rules. Further in this paper we develop an algorithm for frequent pattern discovery in this language. DL-safe rules combine description logics with disjunctive Datalog, which we briefly recall below. 

\subsubsection{Description logics}

Description logics (DLs) \cite{DL_handbook} are a family of knowledge representation languages, specifically suited to represent terminological knowledge in a structured and formalized way. Two kinds of atomic symbols are distinguished in any description logic language: {\em atomic concepts} (denoted by $A$) and {\em atomic roles} (denoted by $R$ and $S$). Atomic symbols are {\em elementary descriptions} from which we inductively build {\em complex descriptions} (denoted by $C$ and $D$) using {\em concept constructors} and {\em role constructors}.  Description logics differ by the set of constructors they admit. 

DLs are equipped with a logic-based model-theoretic semantics. The semantics is defined by {\em interpretations} 
$\mathcal{I} = (\Delta^{\mathcal{I}}, \cdot^{\mathcal{I}})$, where the non-empty set $\Delta^{\mathcal{I}}$ is the domain of the interpretation and the interpretation function $\cdot^{\mathcal{I}}$ assigns a set 
$A^{\mathcal{I}} \subseteq \Delta ^{\mathcal{I}}$ to every atomic concept $A$ and a binary relation 
$R^{\mathcal{I}} \subseteq \Delta ^{\mathcal{I}} \times  \Delta ^{\mathcal{I}} $ to every atomic role $R$.
The interpretation function is extended to concept descriptions by an inductive definition. The syntax and semantics of $\mathcal{SHIF}$ DL is defined in Table \ref{tab:2.1a}.

\begin{table}[h]
\caption{Syntax and semantics of $\mathcal{SHIF}$.}\label{tab:2.1a}
\centering\footnotesize%
\begin{tabular}{l l l l}
\hline
Constructor & Syntax & Semantics \\
\hline
Concept constructors\\
\hline
	Universal concept	& $\top$ &	 $\Delta^\mathcal{I}$\\ 
	Bottom concept	& $\bot$ &	 $\emptyset$\\ 
	Negation	of arbitrary concepts & $(\neg C)$  &  $\Delta^\mathcal{I}$$\backslash$$C^\mathcal{I}$\\ 
	Intersection	& $(C \sqcap D)$  & $C^\mathcal{I} \cap D^\mathcal{I}$ \\
   Union	& $(C \sqcup D)$  & $C^\mathcal{I} \cup D^\mathcal{I}$ \\
	Value restriction	& $(\forall R$.$C)$  & \{$a\in\Delta^\mathcal{I} | \forall b: (a, b)\in R^\mathcal{I} \rightarrow b\in C^\mathcal{I}\}$\\ 
	Full existential quantification	& $(\exists R$.$C)$  &	\{$a\in\Delta^\mathcal{I} | \exists b: (a, b)\in R^\mathcal{I} \wedge b\in C^\mathcal{I}\} $ \\
   Functionality & $\leq$1$R$ &  $\Big\{a \in\Delta^\mathcal{I} \big| \   |\{b | (a, b)\in R^\mathcal{I}\}| \leq 1\Big\}$\\
\hline
Role constructors\\
\hline
Inverse role & $R^{-}$ & $\{(a, b) \in \Delta^\mathcal{I} \times \Delta^\mathcal{I} | (b, a) \in R^\mathcal{I}\}$\\
Transitive role & \textsf{Trans}$(R)$ & $R^\mathcal{I}$ is transitive\\ 
\hline
\end{tabular}
\end{table}

A description logic {\em knowledge base, KB}, is typically divided into an {\em intensional part} ({\em terminological one, TBox}), and an {\em extensional part} ({\em assertional one, ABox}). The TBox contains axioms dealing with how concepts and roles are related to each other ({\em terminological axioms}), while the ABox contains assertions about individuals ({\em assertional axioms}). A semantics is given to ABoxes by extending interpretations $\mathcal{I} = (\Delta^{\mathcal{I}}, \cdot^{\mathcal{I}})$ by an additional mapping of each individual name $a$ to an element $a^{\mathcal{I}} \in \Delta ^{\mathcal{I}}$. The interpretation $\mathcal{I}$ satisfies a set of axioms (a TBox $\mathcal T$ or/and an ABox $\mathcal A$) iff it satisfies each element of this set.

\subsubsection{Disjunctive Datalog}

Disjunctive Datalog \cite{Eiter_DD} is an extension of Datalog that allows disjunctions of literals in the rule heads.

\newtheorem{definition}{Definition}
\begin{definition}[Disjunctive Datalog rule]
\label{DDR}

A disjunctive Datalog rule is a clause of the form

$$H_1 \vee \ldots \vee H_k \leftarrow B_1, \ldots, B_n$$

where $H_i$ and $B_j$ are atoms, and $k \geq  1, n \geq 0$.
$\square$
\end{definition}

\begin{definition}[Disjunctive logic program]
\label{DLP}

A disjunctive logic program $P$ is a finite collection of disjunctive Datalog rules.
$\square$
\end{definition}

We consider only disjunctive Datalog programs without negative literals in the body, that is, {\em positive programs}. 

For the semantics, only Herbrand models are considered, and the semantics of $P$ is defined by the set of all minimal models $M$ of $P$, denoted by $\mathcal{MM}(P)$. 
A ground literal $L$ is called a {\em cautious answer} of $P$, written $P \models_c L$, if $L \in M$ for all $M \in  \mathcal{MM}(P)$. FOL entailment coincides with cautious entailment for positive ground atoms on positive programs.

\subsubsection{DL-safe rules}

We use the formalism of DL-safe rules introduced in \cite{dl_safe_rules2}. 
The description logic $\mathcal{SHIF}$ and disjunctive Datalog rules are integrated by allowing concepts and roles to occur in rules as unary and binary predicates, respectively. 
Below we define DL-safe rules with respect to the description logic $\mathcal{SHIF}$ and disjunctive Datalog rules.

\begin{definition}[DL-safe rules]
\label{DLSR}

Let $KB$ be a knowledge base represened in the $\mathcal{SHIF}$ language. A {\em DL-predicate} is an atomic concept or a simple role from $KB$. For $t_1$ and $t_2$ being constants or variables,  a {\em DL-atom} is an atom of the form $A(t_1)$, where $A$ is an atomic concept in $KB$, or of the form $R(t_1, t_2)$, where $R$ is a simple role in $KB$, or of the form $t_1 = t_2$. A {\em non-DL-predicate} is any other predicate than =, an atomic concept in $KB$, or a role in $KB$. A {\em non-DL-atom} is an atom with any predicate other than =, an atomic concept in $KB$, and a role in $KB$. A (disjunctive) DL-rule is a (disjunctive) rule with DL- and non-DL-atoms in the head and in the body. A (disjunctive) DL-program $P$ is a finite set of (disjunctive) DL-rules. A {\em combined knowledge base} is a pair (KB, P). A (disjunctive) DL-rule $r$ is called {\em DL-safe} if each variable in $r$ occurs in a non-DL-atom in the rule body.  A (disjunctive) DL-program $P$ is DL-safe if all its rules are DL-safe.
$\square$
\end{definition}

In order to define the semantics of a combined knowledge base $(KB, P)$, the $KB$ axioms are mapped into a (disjunctive) Datalog program $DD(KB)$, which entails exactly the same set of ground facts as $KB$. 
The details concerning the mapping can be found in \cite{motik_phd,motik_aboxes}. It is proved that $KB$ is satisfiable with respect to the standard  model-theoretic semantics of $\mathcal{SHIF}$ iff $DD(KB)$ is satisfiable in first-order logic \cite{motik_phd}. It is also proved in \cite{motik_phd} that for a combined knowledge base $(KB, P)$ consisting of a $\mathcal{SHIF}$ knowledge base $KB$ and a finite set of DL-safe rules $P$, $(KB,P) \models  \alpha$  iff $DD(KB) \cup  P \models  \alpha $, for a ground atom $\alpha $, where $\alpha $ is of the form $A(a)$ or $R(a, b)$, and $A$ is an atomic concept. Therefore, reasoning in $(KB, P)$ can be performed using the well-known techniques from the field of deductive databases.

DL-safety implies that each variable is bound only to constants explicitely introduced in a $(KB, P)$. Let us consider, for example, a combined knowledge base $(KB, P)$ such that $KB$ contains the concept {\em Person} and roles {\em livesAt} and {\em worksAt}, while $P$ contains the following rule defining {\em Homeworker} as a person who lives and works at the same place:  

\begin{equation}
\label{regula}
Homeworker(x) \leftarrow Person(x), livesAt(x,y), worksAt(x,y)
\end{equation}

This rule is not DL-safe. It is because the variables $x$ and $y$ that occur in the DL-atoms $Person(x)$, $livesAt(x,y)$, $worksAt(x,y)$ do not occur in the body in any non-DL-atom. Let us introduce a special non-DL-predicate $\mathcal{O}$ such that the fact $\mathcal{O}(a)$ occurs for each individual $a$ in the ABox. In order to make rule (\ref{regula}) DL-safe, we add non-DL atoms $\mathcal{O}(x)$ and $\mathcal{O}(y)$ in the rule body, obtaining:

\begin{equation}
\label{DLSregula} 
Homeworker(x) \leftarrow Person(x), livesAt(x,y), worksAt(x,y), \mathcal{O}(x), \mathcal{O}(y)
\end{equation}

In order to express a DL-safe rule intuitively, we just append to the original rule the phrase: "where the identity of all objects is known". The rule (\ref{DLSregula}) can be intuitively expressed as follows: "A Homeworker is a {\em known} person who lives at and works at the same {\em known} place".

A combined knowledge base $(KB, P)$ may be divided into an {\em intensional part}, which contains knowledge independent of any specific instances, and an {\em extensional part}, which contains factual knowledge. 

\subsection{Problem of onto-relational frequent pattern discovery}

In this subsection we formally define the problem of frequent pattern discovery from knowledge bases represented in the DL-safe rules, as it is addressed in this paper. Initial formulation of this problem has been presented in \cite{orai_ruleml2005}. This subsection specializes it. Let us start with an example of a combined knowledge base $(KB, P)$.

\newtheorem{example}{Example}
\begin{example}[Example knowledge base $(KB, P)$] \label{ex:KB_statement}
Given is a knowledge base $(KB, P)$ describing bank services, presented in Table~\ref{tab:kb}.
For the clarity of presentation, non-DL-predicates are denoted with prefix $p\_$. 
\begin{table}
\small
\begin{tabular*}{1.0\textwidth}[t]{l | p{6.5cm} l }
\hline
\textbf{Terminology in \textit{KB}}\\
\hline
$Client \equiv \exists isOwnerOf$ &	A client is defined as an owner of something. \\
$\top \sqsubseteq \forall isOwnerOf$.$Account \sqcup CreditCard$	&	The range of $isOwnerOf$ is a disjunction of $Account$ and $CreditCard$.  	\\

$\exists isOwnerOf^{-} \sqsubseteq Property$ &	Having an owner means being a property. \\
$Gold \sqsubseteq CreditCard$ &	$Gold$ is a subclass of $CreditCard$.  	\\
	&		 \\
$relative \equiv relative^{-}$ 
  	&	The role $relative$ is symmetric.	\\
   	&	 \\
$Account \sqsubseteq \exists isOwnerOf^{-}$  	&	Each account has an owner. \\
   	&	 \\
$\top \sqsubseteq \forall hasMortgage$.$Mortgage$ &	The range of $hasMortgage$ is $Mortgage$.	\\
$\top \sqsubseteq \forall hasMortgage^{-}$.$Account$  &	The domain of $hasMortgage$ is $Account$.	\\
$\top \sqsubseteq\ \leq 1 hasMortgage^{-}$ &	A mortgage can be associated up to one account.	\\

$Account \equiv \neg CreditCard$  &	$Account$ is disjoint with $CreditCard$.\\
   	&		\\
\hline
\textbf{Rules in \textit{P}}\\
\hline
$p\_familyAccount(x,y,z)\leftarrow Account(x),$ & $p\_familyAccount$ is an account that is \\
\hspace{15pt} $isOwner(y,x), isOwner(z,x), $ & co-owned by relatives.\\
\hspace{15pt}  $relative(y,z), \mathcal{O}(x),\mathcal{O}(y),\mathcal{O}(z)$ & \\
& \\ 
$p\_sharedAccount(x,y,z) \leftarrow$  	&  Family account is a shared account.   \\
\hspace{15pt}  $p\_familyAccount(x,y,z)$ &\\
 & \\  	 
$p\_man(x) \vee p\_woman(x) \leftarrow Client(x), \mathcal{O}(x)$ 	& A client is a man or a woman.\\
\\
\hline
\textbf{Assertions in \textit{(KB, P)}}\\
\hline
$p\_woman(Anna)$ & $Anna$ is a woman.\\
$isOwnerOf(Anna,a1)$ & $Anna$ is an owner of $a1$.\\
$hasMortgage(a1,m1)$ & Mortgage $m1$ is associated to account $a1$.\\ 

$relative(Anna,Marek)$ & $Anna$ is a relative of $Marek$.\\
   	&		\\
$isOwnerOf(Jan,cc1)$ & $Jan$ is an owner of $cc1$.\\
$CreditCard(cc1)$ & $cc1$ is a credit card.\\
   	&		\\ 
$isOwnerOf(Marek,a1)$ & $Marek$ is an owner of $a1$.\\
   	&		\\
$Account(account2)$ & $account2$ is an account.\\
   	&		\\
$\mathcal{O}(i)$ for each explicitly named individual $i$ & Enumeration of all ABox individuals\\
\hline
\end{tabular*} 
\caption{An example of a combined knowledge base.}
\label{tab:kb}
\end{table}
\normalsize
This knowledge base could not be represented in description logic or Datalog alone. Definite Horn rules require all variables to be universally quantified, and therefore it is impossible to assert the existence of unknown individuals. For example, it is impossible to assert that each account must have an owner. Moreover, Horn rules are unable to represent disjunctions in rule heads, and hence, it is not possible to model that the range of the role $isOwnerOf$ is a disjunction of $Account$ and $CreditCard$. In description logic, in turn, it is not possible to define a "triangle" relationship that is modelled by the rule defining $p\_familyAccount$. $\square$ 

\end{example}

The task addressed in this paper is frequent pattern discovery. The patterns being found in our approach have the form of {\em conjunctive queries} over the combined knowledge base $(KB, P)$. An answer set of a query contains individuals of a user-specified reference concept \textit{\^C}. We assume that the queries are positive, i.e. they do not contain any negative literals. Moreover, we assume that the queries are DL-safe. This means that all variables in a query are bound to instances explicitly occurring in $(KB, P)$, even if they are not returned as a part of the query answer. In this context a query is defined as follows.

\begin{definition}[Conjunctive DL-safe queries]
\label{def:query}

Let $(KB, P)$ be a combined knowledge base in DL-safe rules with $KB$ represented in $\mathcal{SHIF}$. Let ${\bf x} = \{x_1, \ldots, x_m\}$ be a set of undistinguished variables (the variables whose bindings are not a part of the answer) and {\em key} be the only distinguished variable (that is the variable whose bindings are returned in the answer). A {\em conjunctive query} $Q(key,{\bf x})$ over $(KB,P)$  is a rule using a special predicate name (that does not belong to the set of names occurring in $(KB, P)$) in the head, and whose body is a finite conjunction of atoms of the form $B(t_1,\ldots, t_n)$, where $B$ is an $n$-ary predicate (either from the $KB$ component or from the disjunctive Datalog program $P$) and $t_i, i=1, \ldots, n$, is the distinguished variable {\em key} or a variable from ${\bf x}$.
A conjunctive query $Q(key,{\bf x})$ is {\em DL-safe} if each variable occurring in a DL-atom also occurs in a non-DL atom in $Q(key,{\bf x})$.

The inference problems for conjunctive queries are defined as follows:
\begin{itemize}
\item {\em Query answering}: An answer to a query $Q(key,{\bf x})$ w.r.t. $(KB,P)$ is an assignment $\theta $ of an individual to the distinguished variable $key$ such that $(KB,P) \models \exists {\bf x} : Q(key\theta, {\bf x} ) $.
\item {\em Query containment}: A query $Q_2(key,{\bf x}_2)$ is contained in a query $Q_1(key, {\bf x}_1)$ w.r.t. $(KB, P)$ if $(KB,P) \models \forall key : [\exists {\bf x}_2: Q_2(key,{\bf x}_2) \rightarrow \exists {\bf x}_1 : Q_1(key,{\bf x}_1)$].$\square$ 
\end{itemize}
\end{definition}

In our approach patterns are positive (i.e., without negative literals) conjunctive DL-safe queries over the combined knowledge base $(KB, P)$ addressing a user-specified reference concept \textit{ \^ C}.  The atom with a reference concept as the predicate contains the only distinguished variable $key$.

\begin{definition}[Pattern]
\label{def:pattern}

Given is a combined knowledge base $(KB, P)$. A pattern $Q$ is a conjunctive, positive DL-safe query over $(KB, P)$ of the following form:

$Q(key) = ?- $\textit{ \^ C}$(key), B_1, \ldots, B_n, \mathcal{O}(key), \mathcal{O}(x_1), \ldots, \mathcal{O}(x_m)$

\noindent where $B_1, \ldots, B_n$ represent atoms of the query, $Q(key)$ denotes that variable $key$ is the only distinguished variable, and $x_1, \ldots, x_m$ represent the undistinguished variables of the query. $Q(key)$ is called the \emph{head} of $Q$, denoted $head(Q)$, and the conjunction \textit{ \^ C}$(key), B_1, \ldots, B_n, \mathcal{O}(key), \mathcal{O}(x_1), \ldots, \mathcal{O}(x_m)$ is called the \emph{body} of $Q$, denoted $body(Q)$. A trivial pattern is the query of the form: $Q(key)=?- $\textit {\^ C}$(key), \mathcal{O}(key)$.
$\square$
\end{definition}

We assume each query posseses the \textit{linkedness} property, that is each variable in the body of a query is linked to the variable {\textit key} through a path of atoms. 

\begin{definition}[Linkedness]
A variable $x$ is \emph{linked} in a query $Q$ iff $x$ occurs in the head of $Q$ or there is an atom $B$ in the body of $Q$ that contains the variable $x$ and a variable $y$ (different from $x$), and $y$ is linked.
$\square$
\end{definition}

Examples of patterns that can be discovered from the knowlegde base introduced in Example \ref{ex:KB_statement} are presented below.

\begin{example}[Example patterns]\label{ex:patterns}
Consider the knowledge base $(KB, P)$ from Example \ref{ex:KB_statement}. Assuming that $Client$ is the reference concept $\hat{C}$, the following patterns over $(KB, P)$, may be built:\\\\
\ \ $Q_{ref}(key) = ?- Client(key), \mathcal{O}(key)$\\
\ \ $Q_1(key) = ?- Client(key), isOwnerOf(key, x), \mathcal{O}(key),\mathcal{O}(x)$\\
\ \ $Q_2(key) = ?- Client(key), isOwnerOf(key, x), p\_familyAccount(x, key, z)$\\
\ \ $Q_3(key) = ?- Client(key), isOwnerOf(key, x), isOwnerOf(key, y),\mathcal{O}(key),\mathcal{O}(x),\mathcal{O}(y)$\\
\ \ $Q_4(key) = ?- Client(key), isOwnerOf(key, x), CreditCard(x),\mathcal{O}(key),\mathcal{O}(x)$\\
\\
where $Q_{ref}$ is a \emph{reference query}, counting the number of instances of $\hat{C}$. $\square$ 
\end{example}

In order to define the task of frequent pattern discovery we need to define how to calculate the pattern support.

\begin{definition}[Support]
\label{support}
Let $Q$ be a query over a combined knowledge base $(KB, P)$, $answerset(\textit{\^ C}, Q, (KB, P))$ be a function that returns the set of all instances of concept \textit{\^ C} that satisfy query $Q$ with respect to $(KB, P)$, and let $Q_{ref}$ denote a trivial query for which the answerset contains all instances of the reference concept \textit{\^ C} in $(KB, P)$.  
\\
A support of query $Q$ with respect to the knowledge base $(KB, P)$ is defined as the ratio between the number of instances of the reference concept \textit{\^ C} that satisfy query $Q$ w.r.t. $(KB, P)$ and the total number of instances of the reference concept \textit {\^ C}:
$$support(\textit{\^ C}, Q, (KB, P)) = \frac{|answerset(\textit{\^ C}, Q, (KB, P))|}{|answerset(\textit{\^ C}, Q_{ref}, (KB, P))|}$$
$\square$
\end{definition}

The support is calculated as the ratio of the number of bindings of variable $key$ in the given query $Q$ to the number of bindings of variable $key$ in the reference query $Q_{ref}$. The reference concept \textit {\^ C} determines what is counted. Let us now calculate the support of query $Q_2$ from Example \ref{ex:patterns}.

\begin{example}
For the illustration of the support notion, consider the queries from Example \ref{ex:patterns}. The reference query has 3 items in its answer set that is 3 individuals from $(KB, P)$ that are deduced to be $Client$ due to the axiom defining a client as an owner of something. Query $Q_2$, for example, has 2 items in its answer set that is the clients that are co-owners of at least one account with their relatives ($Anna$, $Marek$). The support of query $Q_2$ is then calculated as: $support(\hat{C}, Q_2, (KB, P)) = \frac{2}{3} \approx $0.66.$\square$
\end{example}

Finally, we can formulate our task of frequent pattern discovery in a combined knowledge base $(KB, P)$.

\begin{definition}[Frequent pattern discovery]
\label{FPD}

Given
\begin{itemize}
\item a combined knowledge base $(KB, P)$ represented in DL-safe rules, where $KB$ is represented in $\mathcal{SHIF}$ and $P$ is a positive disjunctive Datalog program,
\item a set of patterns in the form of queries $Q$ that all contain a reference concept \textit{\^ C} as a predicate in one of the atoms in the body and where the variable in the atom \textit{\^ C} is the only distinguished variable,
\item a minimum support threshold \textit{minsup} specified by the user,
\end{itemize}
and assuming that queries with support $s$ are frequent in $(KB,P)$ if $s \geq minsup$, the task of {\em frequent pattern discovery} is to find the set of frequent queries.
$\square$
\end{definition}

\begin{example}
Let us assume the threshold $minsup$=0.5 and let us consider the queries from Example \ref{ex:patterns}. The set of frequent patterns is then $\{Q_{ref}, Q_1, Q_2, O_3\}$.
$\square$
\end{example} 
\section{Solution algorithm}
\label{sec:algorithm}
The main contribution of this paper is the algorithm for frequent pattern discovery in combined knowledge bases represented in DL-safe rules as described in Section \ref{sec:prelim}. Initial results on the algorithm development have been presented in \cite{orai_ekaw2006,orai_rr2008}. This section advances them.
Our method follows the usual approach where the search starts with the most general patterns and refines them to more specific ones in consecutive steps. Thus, firstly, we define the generality relation and further the refinement operator that computes a set of specializations of a pattern. 

\subsection{Generality relation}

We use a semantic generality relation in order to fully utilize the information stored in the combined knowledge base $(KB,P)$. As we have defined in Section \ref{sec:prelim}, patterns are represented as queries, so it seems natural to define the generality relation as the {\em query containment} (or {\em subsumption}) relation. 

\begin{definition}[Generality relation]
\label{GR}
Given two patterns $Q_1$ and $Q_2$ defined as queries over a combined knowledge base $(KB, P)$ (see Definition \ref{def:pattern}) we say that pattern $Q_1$ is at least as general as pattern $Q_2$ under query containment w.r.t. $(KB, P)$, $Q_1  \succeq_{\mathcal B} Q_2$, iff query $Q_2$ is contained in query $Q_1$ w.r.t. $(KB, P)$.
$\square$
\end{definition}

Theorem \ref{th:testing_containment} lays the foundations for an algorithm to test the pattern subsumption. 
\begin{theorem}\label{th:testing_containment}[Testing $\succeq_{\mathcal B}$]
Let $Q_1(key,{\bf x_1})$ and $Q_2(key,{\bf x_2})$ be two queries and $(KB, P)$ be a combined knowledge base. Let $\theta$ be a substitution grounding the variables in $Q_2$ using new constants not occuring in $(KB, P)$ (Skolem substitution). 
Then $Q_1  \succeq_{\mathcal B} Q_2$ if and only if there exists a ground substitution $\sigma$ for $Q_1$ such that
\begin{description}
\item[(i)] $head(Q_2)\theta=head(Q_1)\sigma$ and
\item[(ii)] $(KB, P) \cup body(Q_2)\theta \models body(Q_1)\sigma$
\end{description}
\end{theorem}
\begin{proof}
($\Leftarrow$) Assume there exists a ground substitution $\sigma$ for $Q_1$ such that $(i)$ and $(ii)$. Let $a$ be some individual, $\mathcal{I}_a$ be some interpretation of $(KB, P)$ which is a model of $(KB, P)$ such that $a$ is an answer to the query $Q_2$ in $\mathcal{I}_a$. In order to prove that $Q_1  \succeq_{\mathcal B} Q_2$ we need to prove that $a$ is also an answer to the query $Q_1$ in $\mathcal{I}_a$. By definition of query answering (Definition \ref{def:query}) there exists a substitution $\phi$ such that $a$ is identical to $key\phi$ and  $\exists body(Q_2)\phi$ is true in $\mathcal{I}_a$. Since $Q_2$ is DL-safe there must exist another substitution, $\phi '$, such that $Q_2\phi '$ is ground, $a$ is indentical to $key\phi '$ and $body(Q_2)\phi '$ is true in $\mathcal{I}_a$. Because formula $(KB, P) \cup body(Q_2)\theta \models body(Q_1)\sigma$ is valid, by the uniform replacement of constants we have $head(Q_2)\phi '=head(Q_1)\sigma$ and $(KB, P) \cup body(Q_2)\phi ' \models body(Q_1)\sigma$, so $body(Q_1)\sigma$ is also true in $\mathcal{I}_a$. Because $head(Q_2)\phi '$ is identical to $head(Q_1)\sigma$ this implies $a=key\phi '$ is an answer to the query $Q_1$. This argument follows for any interpretation $\mathcal{I}$ satisfying the initial constraints, so $Q_1  \succeq_{\mathcal B} Q_2$.
\\
($\Rightarrow$) Assume $Q_1  \succeq_{\mathcal B} Q_2$. The following arguments show that a ground substitution $\sigma$ exists. Let a substitution $\theta$ be given as in the theorem. Let $\mathcal{I}$ be a model of $(KB, P) \cup Q_2\theta$. Since $key\theta$ is an answer to $Q_2$ in $\mathcal{I}$,  $key\theta$ is also an answer to $Q_1$ in $\mathcal{I}$. Moreover, since $Q_1$ is DL-safe, there must exist a ground substitution $\phi$ such that $head(Q_2)\theta=head(Q_1)\phi$, and $body(Q_1)\phi$ is true in $\mathcal{I}$. By the uniform replacement of constants we obtain that $head(Q_2)\theta=head(Q_1)\sigma$, and $body(Q_1)\sigma$ is true in $\mathcal{I}$.
This argumentation is valid for any interpretation satisfying the constraints, so the thesis follows.    
\end{proof}

Below we prove that appending an atom to a query results in an equally or more specific query which gives an easy way to building specializations of a query.

\begin{proposition}\label{prop:cor}
Let $Q_2$ be a query over $(KB, P)$, built from query $Q_1$ by adding an atom. It holds that $Q_1$ $\succeq_\mathcal{B}$ $Q_2$. 
\end{proposition}

\begin{proof}
Let us consider query $Q_1 =?- \hat C(key), B_1, \ldots, B_n$ and let us add atom $B_{n+1}$ to $Q_1$ obtaining query $Q_2 =?- \hat C(key), B_1, \ldots, B_n, B_{n+1}$. Let $\theta$ be an answer to query $Q_2$. According to Definition \ref{def:query}, $(KB,P) \models \exists {\bf x} : Q_2(key\theta, {\bf x} ) $. But since a query is a conjunction of atoms, it follows that also $(KB,P) \models \exists {\bf x} : Q_1(key\theta, {\bf x} ) $. Thus, the answer set of query $Q_2$ is a subset of the answer set of query $Q_1$ what completes the proof.
\end{proof}

A crucial property of the generality relation that allows to develop efficient algorithms is monotonicity with regard to support.

\begin{proposition}{}
\label{prop:mono}

Let $Q_1$ and $Q_2$ be two queries over the combined knowledge base $(KB, P)$ that both contain the reference concept $\textit{\^ C}$. If $Q_1  \succeq_{\mathcal B} Q_2$ then $support(\textit{\^ C}, Q_1, (KB, P)) \geq support(\textit{\^ C}, Q_2, (KB, P))$.
\end{proposition}

\begin{proof}
If $Q_1  \succeq_{\mathcal B} Q_2$ then, by Definition \ref{GR}, query $Q_2$ is contained in query $Q_1$. Further, from Definition \ref{def:query} we conclude that since query $Q_2$ is contained in query $Q_1$ then for any possible  extensional part of $(KB, P)$, while keeping the same intensional part, the answer set of $Q_2$ is contained in the answer set of $Q_1$, and in consequence by Definition \ref{support},  $support(\textit{\^ C}, Q_1, (KB, P)) \geq support(\textit{\^ C}, Q_2, (KB, P))$, what completes the proof.
\end{proof}

The monotonicity of the query containment with regard to the query support means that none of the specializations of an infrequent pattern can be frequent. 

The generality relation $\succeq_{\mathcal B}$ is a reflexive and transitive binary relation, and so it is a {\em quasi-order} on the space of patterns. It is known \cite{NC_Wolf:1997} that any quasi-ordered space may be searched using refinement operators. 
In the next section we define the refinement operator used in our algorithm.

\subsection{Refinement operator}
\label{subsec:ref}

We define a downward refinement operator that computes a set of specializations of a query. This set is obtained using both syntax and semantics of the query. Firstly, a query is appended with a single atom according to the rules given in Definition \ref{def:ref_rules}. In the second step semantic tests are performed which may exclude further patterns from consideration.

It is convenient to represent the results of the refinement steps on a special {\em trie} structure that was introduced in the \textsc{FARMER} method \cite{FARMER:2001,FARMER:2003}. Trie is a tree with nodes corresponding to the atoms of the query, so that each path from the root to any node corresponds to a query. In consequence every node in a trie defines a query. According to Propositions \ref{prop:cor} and \ref{prop:mono} only nodes which correspond to frequent queries need to be expanded further. 

An example of the trie data structure for a data mining problem defined over the knowledge base from Example \ref{ex:KB_statement} is presented in Figure~\ref{fig:trie}. In order to build the patterns the following predicates from the knowledge base were selected: $Client$, $isOwnerOf$, $relative$, $hasMortgage$ and $p\_woman$. Notice that the special purpose predicates (the ones of the form $\mathcal{O}(x)$) are omitted from the presentation in the trie. The presence of such predicates indicates that a query is DL-safe. As we assume that all queries within our approach are DL-safe, we can omit the special purpose predicates for simplicity. The superscripts in Figure \ref{fig:trie} correspond to the two ways described in Definition 
\ref{def:ref_rules} in which atoms are added to the query.

\begin{figure}[t]
\includegraphics[width=12cm]{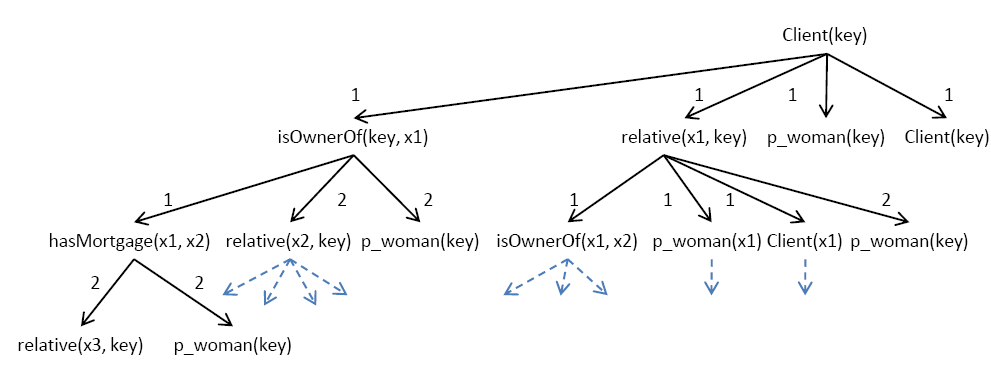}
\caption{A part of the trie constructed for the $(KB, P)$ from Example \ref{ex:KB_statement}, $\hat{C}$=$Client$, $minsup$ = 0.2.}
\label{fig:trie}

\end{figure}

\begin{definition} \label{def:ref_rules} 
Let $T$ be a trie data structure that imposes an order of atoms in a query. Let $Q$ be a query, let $B$ be $last(Q)$ that is the last atom in query $Q$, let $B_p$ be the parent of $B$ in $T$. A variable is called \emph{new} if it does not occur in any earlier atom of a query. Atoms are added to trie $T$ as:
\begin{enumerate}
\item \emph{dependent atoms} (share at least one variable with $last(Q)$, that was new in $last(Q)$), 
\item \emph{right brothers} of a given node in $T$ (these are the copies of atoms that have the same parent $B_p$ as the given atom $B$ and are placed on the right-hand side of $B$ in $B_p$'s child list), new variables are renamed such that they are also new in the copy.$\square$
\end{enumerate}
\end{definition}

The first rule introduces the dependent atoms that could not be added earlier. The dependent atoms are brothers of each other in the trie. The second rule, the right brother copying mechanism, takes care that all possible subsets but only one permutation out of the set of dependent atoms is considered.

Let us now introduce the semantic tests which are performed as the second step of the refinement procedure in order to reduce the set of patterns submitted for frequency evaluation. Due to the efficiency reasons semantic tests are performed on a knowledge base $(KB, P)$ with ground facts like concept and role assertions removed. This reduced knowledge base is denoted by $cp(KB, P)$. The first test consists in determining the query satisfiability, further ones check for some kinds of semantic redundancy as described later in this section.

The test for checking satisfiability of query $Q(key, {\bf x})$ with regard to knowledge base $cp(KB, P)$ consists in checking whether $cp(KB, P) \cup \{ \exists {key, \bf x} : Q$ \} is satisfiable that is  whether there is a model of $cp(KB, P)$ in which there is some valuation for the distinguished variable $key$ and undistinguished variables {\bf x}. The variables are skolemized, and assuming that $a$ and {\bf b} are new constants, $Q(a, \bf b)$ is asserted to $cp(KB, P)$. Then it is checked whether the updated $cp(KB, P)$ is satisfiable. The query satisfiability test described above is defined in Definition \ref{def:query_SAT} below. 

\begin{definition} \label{def:query_SAT}
Query $Q$ is \emph{satisfiable} w.r.t. a combined knowledge base $(KB, P)$ iff 
$(KB, P) \cup Q\theta $ is satisfiable, where $\theta$ is a Skolem substitution. 
$\square$
\end{definition} 

\begin{example}
Let us consider the knowledge base $(KB, P)$ from Example \ref{ex:KB_statement} and the query: 
$$Q(key) =? - Account(key), CreditCard(key), \mathcal{O}(key)$$

Since in $(KB, P)$ the concepts $Account$ and $CreditCard$ are specified as disjoint, we know a priori that it is useless to submit the query $Q$ as it cannot have any answer due to its unsatisfiability. $\square$
\end{example}

After performing the satisfiability test, the queries are further pruned in order to obtain only those candidates that are not {\em semantically redundant}. We consider two kinds of semantic redundancy. The first kind occurs when a query has redundant atoms, that is atoms that can be deduced from other atoms in the query. The second kind occurs when there are frequent queries already found in the earlier steps that are \emph{semantically equivalent} to the newly generated candidate. 

\indent
In order to avoid the first kind of redundancy, the queries are tested for \emph{semantic freeness}\index{semantic freeness, s-freeness}. Only \emph{semantically free} queries are kept for further processing. The notion of the semantic freeness has been introduced in \cite{carmr}. It is adapted to our setting as follows. 

\begin{definition} [Semantically free pattern]\label{def:sfree_pattern}
A pattern $Q$ is \emph{semantically free} or \emph{s-free} w.r.t a combined  knowledge base $(KB, P)$ if there is no pattern $Q'$, built from $Q$ by removing any atom, such that $Q \succeq_\mathcal{B} Q'$.
$\square$
\end{definition}

\begin{example}
Given is the knowledge base from Example \ref{ex:KB_statement} and the following queries to this knowledge base:\\
\\
\ \ $Q_1(key) = ?- Account(key), isOwnerOf(x, key), \mathcal{O}(key),\mathcal{O}(x)$\\
\ \ $Q_2(key) = ?- Account(key), isOwnerOf(x, key), Client(x), \mathcal{O}(key),\mathcal{O}(x)$\\
\\
Query $Q_1$ is s-free while query $Q_2$ is not. The reason why the second query is not s-free is that atom $Client(x)$ can be deduced from the other atoms of this query. More specifically, the atom $Client(x)$ can be deduced from the atom $isOwnerOf(x, key)$ as from the axioms in the knowledge base it follows that any object being asserted to the domain of $isOwnerOf$ is a $Client$. $\square$
\end{example} 

Moreover, the test for semantic freeness is performed on a query with the atom \textit{\^ C(key)} removed. It is motivated by the fact, that some queries could be pruned after the s-freeness test, that we do not necessarily would like to be pruned, just because of the obligatory presence of the reference concept in each query. 

\begin{example}
Consider the knowledge base from Example \ref{ex:KB_statement} and $Client$ as a reference concept. Then query:
$$Q(key) = ?- Client(key), isOwnerOf(key, x), \mathcal{O}(key),\mathcal{O}(x)$$
does not pass the s-freeness test from Definition \ref{def:sfree_pattern}, because the atom $Client(key)$ can be deduced from the second atom of $Q$. However, the atom $Client(key)$ contains the reference concept, which is obligatory in each query. Consider now query $Q'$, obtained by removing the atom with a reference concept from query $Q$:
$$Q'(key) = ?- isOwnerOf(key, x), \mathcal{O}(key),\mathcal{O}(x)$$
The modified query, $Q'$, is s-free. Reconsider the queries from Example \ref{ex:patterns}. Queries $Q_{ref}$, $Q_1$, $Q_2$ and $Q_4$ are s-free with regard to the modified s-freeness test, while query $Q_3$ is not s-free. $\square $
\end{example}

A candidate query may be semantically redundant not only due to redundant atoms. The second kind of redundancy occurs when a candidate query is semantically equivalent to a frequent one already found. Such patterns are also pruned, which is performed by searching the trie for a pattern equivalent to the given one. So-called \emph{optimal} refinement operator assures that no pattern is generated twice. Using the trie data structure and pruning the candidate patterns that are semantically equivalent to the ones already found, make our refinement operator \emph{optimal}. By pruning semantically equivalent patterns we achieve also the property of \emph{properness} of the refinement operator, that is every pattern $Q'$ generated by the refinement operator is more specific than the pattern $Q$ being refined ($Q'$ is never equivalent to $Q$).

\subsection{The algorithm} 

The approach proposed in this paper follows the common scheme of algorithms for finding frequent patterns which is a ''generate-and-test'' approach. In such approach candidate queries are repeatedly generated and tested for their frequency. In order to generate candidates, a refinement operator is applied. 
\\
\indent
The proposed, recursive node expansion algorithm is presented below. A node being expanded is denoted by $n_i$, $Q(key, \mathbf{x})$ denotes a query with $\mathbf{x}$ being undistinguished variables, $d$ denotes the depth of the current node in the trie $T$. The trie is generated up to the user-specified MAXDEPTH depth.

\begin{small}
\newtheorem{alg}{Algorithm}
\begin{alg}\label{alg:expandNode}
\textbf{expandNode}($n_d$, $Q(key,\mathbf{x})$, $d$, $T$, MAXDEPTH)
\begin{enumerate}
\item
 \textbf{if} $d$ $<$ MAXDEPTH \textbf{then}
\item 
	\hspace{3 mm} \textbf{while} all possible children of $n_d$ not constructed \textbf{do}
\item 
	\hspace{6 mm} construct child node $n_{d+1}$ and associated query $Q_c(key, \mathbf{x})$ using the trie data structure $T$ and refinement rules from Definition \ref{def:ref_rules}
 \item 
	\hspace{6 mm} \textbf{if} $Q_c(key, \mathbf{x})$ is satisfiable wrt $(KB, P)$ \textbf{then}
\item 
\hspace{9 mm} \textbf{if} $Q_c(key, \mathbf{x})$ is semantically free wrt $(KB, P)$ \textbf{then}
\item 
\hspace{12 mm} \textbf{if} $Q_c(key, \mathbf{x})$ is not semantically equivalent wrt $(KB, P)$ to any frequent query found earlier \textbf{then}
\item 
\hspace{15 mm} evaluate candidate query $Q_c(key, \mathbf{x})$
\item 
\hspace{15 mm} \textbf{if} $Q_c(key, \mathbf{x})$ is frequent \textbf{then}
\item 
\hspace{18 mm} addChild($n_d$, $n_{d+1}$); //\textit{add $n_{d+1}$ as a child of $n_d$}
\item 
\hspace{18 mm} $T \leftarrow T \cup n_{d+1}$; 
\item  
\hspace{3 mm} \textbf{for} all children $n_{d+1}$ of node $n_d$ \textbf{do}
\item 
\hspace{6 mm} expandNode($n_{d+1}$, $Q_c(key, \mathbf{x})$, $d+1$, $T$, MAXDEPTH)
\end{enumerate}
\end{alg}
\end{small}

\paragraph{Completeness of search}
Below we prove the \emph{completeness} of our method for pattern refinement, that is we prove that the proposed approach to pattern mining generates for each pattern $Q$ from the space of valid patterns a valid pattern $Q'$ such that $Q'$ is semantically equivalent to $Q$. Valid patterns are those, from the ones defined in Definition \ref{def:pattern}, that are linked and semantically free. 
In order to prove completeness, we relate to the work on FARMER \cite{FARMER:2003}, that originally used trie data structure for relational, frequent pattern mining. 
\\
\indent
First we prove that pruning semantically equivalent patterns (after s-freeness test or after the search on the trie) does not exclude adding all possible refinements to a pattern. 

\begin{lemma}\label{lem:linked} 
Let $(KB, P)$ be a combined knowledge base, and $Q_1$, $Q_2$ be two semantically equivalent patterns ($Q_1 \equiv_{\mathcal B} Q_2$) over $(KB, P)$. Then for each variable $x$ in $Q_1$ there exists a corresponding variable $x'$ in $Q_2$ to which the same bindings can be made as to the variable $x$. 
\label{pr:linkedness}
\end{lemma}

\begin{proof}
By definition (Definition \ref{def:pattern}) both patterns have the same distinguished variable $key$, so the thesis follows for $x=key$. Let us now provide the following argumentation for $x$ being an undistinguished variable.  
Since $Q_1 \equiv_{\mathcal B} Q_2$ then also $Q_1 \succeq_{\mathcal B} Q_2$. Suppose $\theta$ is a Skolem substitution grounding variables in $Q_2$ that satisfies the constraints from Theorem \ref{th:testing_containment}. 
By definition, the substitution $\theta$ assigns a new individual $a$ to variable $key$. The individual $a$ is an answer to $Q_2$, and since $Q_1 \succeq_{\mathcal B} Q_2$, $a$ is also an answer to $Q_1$.  
For $Q_1 \succeq_{\mathcal B} Q_2$ to be valid there must exist a grounding substitution $\sigma$ for $Q_1$ that satisfies the constraints from Theorem \ref{th:testing_containment}. Since $Q_2$ is linked, that is all of its variables are linked to the variable $key$, then also all the constants introduced by $\theta$ are linked to the individual $a$. Since $Q_1$ is linked, then all the constants that bind to variables of $Q_1$ to prove the answer $a$ have to be linked to $a$ as well. Since $a$ and all the constants introduced to the $(KB, P)$ by $\theta$ are new, then any other constants in the $(KB, P)$ are not linked to $a$. In consequence, only the constants introduced to the $(KB, P)$ by $\theta$ can be a part of the substitution $\sigma$. 
Then for each variable $x$ in $Q_1$ there must exist a constant $b$ that is assigned to $x$ by the substitution $\sigma$ and has been introduced by $\theta$. That is why there must exist a variable $x'$ in $Q_2$ for which $\theta$ introduces $b$, and what follows the same bindings that can be made to variable $x$ in $Q_1$ can be as well made to the corresponding variable $x'$ in $Q_2$.
This argumentation is valid for any variables $x$ and $x'$, what completes the proof.
\end{proof}
The following corollary is a consequence of Lemma \ref{lem:linked}.

\begin{corollary}\label{cor:linked}
Let $(KB, P)$ be a combined knowledge base, and $Q_1$, $Q_2$ be two semantically equivalent patterns ($Q_1 \equiv_{\mathcal B} Q_2$) over $(KB, P)$. Then for each variable $x$ in $Q_1$ there exists a corresponding variable $x'$ in $Q_2$ such that any atom $B$ that can be linked to $Q_1$ through the variable $x$ can be also linked to $Q_2$ through the variable $x'$.
\end{corollary}

Subsequently we prove that all possible refinements of a pattern are generated.

\begin{lemma}\label{lem:complet}
Given is a trie $T$, recursively generated by Algorithm \ref{alg:expandNode}, a query $Q$ which occurs in $T$, and an atom $B \notin Q$ which is a valid refinement of $Q$. Then either: 
\begin{description}
\item[(i)] valid query $Q' = (Q_1, B, Q_2)$ exists in trie $T$, for some subdivision of $Q$ into $Q_1$ and $Q_2$, such that $Q = (Q_1,Q_2)$ or
\item[(ii)] valid query $Q''$ exists in trie $T$, such that query $Q''$ is semantically equivalent to query $Q'$.
\end{description}
\end{lemma}

\begin{proof}
Consider case (i). As $B$ is a valid refinement of $Q$, there is a prefix $(Q_p,B_p)$ of $Q$ such that atom $B$ is a dependent atom of $B_p$. 
If $B_p$ is the last atom of $Q$, then it is clear that $B$, as a dependent atom of $B_p$, is generated as a refinement of $Q$ to be added at the end of the query. Dependent atom $B$ is generated by the first rule from Definition \ref{def:ref_rules} and checked for its validity (satisfiability and s-freeness). Hence, query $Q'$ is generated. Let us assume now that $B_p$ is not the last atom and it has different successor $B_{p+1}$ in query $Q$. Atom $B_{p+1}$ is also a child of $B_p$ in $T$. Then let us consider the order of $B$ and $B_{p+1}$ in the list of children of $B_p$ in trie $T$, which is one of the following:
\begin{itemize}
\item $B$ occurs before $B_{p+1}$; then $B_{p+1}$ is a right-hand brother of $B$. The right brothers copying mechanism, the second rule from Definition \ref{def:ref_rules}, will copy $B_{p+1}$ as a child of $B$; the same operations that created $Q$ will create query $Q'$ in subsequent steps.
\item $B$ occurs after $B_{p+1}$; $B$ is copied as a child of $B_{p+1}$. In order to determine the exact injection place of $B$, we recursively apply our arguments, taking into account $B_{p+1}$ and $B$.
\end{itemize}
It follows from the above arguments that query $Q'$ is always generated. After generation of query $Q'$, it is checked, in line 6 of Algorithm \ref{alg:expandNode}, if query $Q'$ is semantically equivalent to some query $Q''$, already present in the trie $T$. If it is the case, $Q''$ is kept in $T$, and $Q'$ is not added to $T$. Otherwise, the newly generated query $Q'$ is added to the trie $T$. Thus, either query $Q'$ exists in the trie $T$ or it is semantically equivalent to query $Q''$. This completes the proof.
\end{proof}

Finally we prove the completeness.

\begin{theorem} [Completeness] For every valid, frequent query $Q_1$ in the pattern space, there is semantically equivalent valid query $Q_2$ in the trie $T$. 
\label{th:completness}
\end{theorem}

\begin{proof} 
Let us assume that queries are generated up to the user specified length (MAXDEPTH). For query $Q_1$ of length 1 it is obvious that there is a corresponding query $Q_2$ of the form $Q(key) = ?- \hat{C}(key), \mathcal{O}(key)$ in the root of the trie (atoms of the form $\mathcal{O}(x)$ are not taken into account as described earlier).   
For query $Q_1$ of length $\geq 1$, the proof is by induction on the length of the query. Assume that an equivalent query for $Q_1\backslash last(Q_1)$ exists in trie $T$. From Corollary \ref{cor:linked} follows that any refinement that can be made to $Q_1\backslash last(Q_1)$ can be also made to any of its equivalent queries. If atom $last(Q_1)$ is a valid refinement of the equivalent query, Lemma \ref{lem:complet} applies. Hence, the thesis follows by induction.
\end{proof} 

\subsection{Implementation}
The proposed method employs several reasoning services run over a combined knowledge base $(KB, P)$ such as: (conjunctive) query answering, deciding knowledge base satisfiability, deciding concept subsumption, classifying the concept hierarchy. In order to perform all these reasoning services, specialized and complex algorithms are needed. As the implementation of such reasoning services is out of the scope of this work, to test our ideas we decided to use an external reasoner \textsc{KAON2}\footnote{http://kaon2.semanticweb.org}. 
\\
\indent
In the core of \textsc{KAON2} there is an algorithm for reducing a DL
knowledge base $KB$ into a disjunctive Datalog program
$DD(KB)$ on which the actual reasoning is performed using the techniques of deductive databases. In particular, \textsc{KAON2} uses a version of Magic Sets optimization technique, originally defined for non-disjunctive programs and recently extended to disjunctive Datalog, in order to identify the part of the database relevant to the query. And it applies semi-na\"{\i}ve, bottom-up evaluation strategy, in order to avoid redundant computation of the same conclusions. Employing these techniques makes
\textsc{KAON2} well suited for a frequent pattern mining application. It has been experimentaly shown that in case of the knowledge bases with relatively small intensional part, but large number of instances, \textsc{KAON2} outperforms
the reasoners using the classical tableaux algorithms by one to two orders of magnitude \cite{motik_aboxes,Pellet-DL-safe:2006}.
\begin{figure}[t]
\centering\includegraphics[width=8cm]{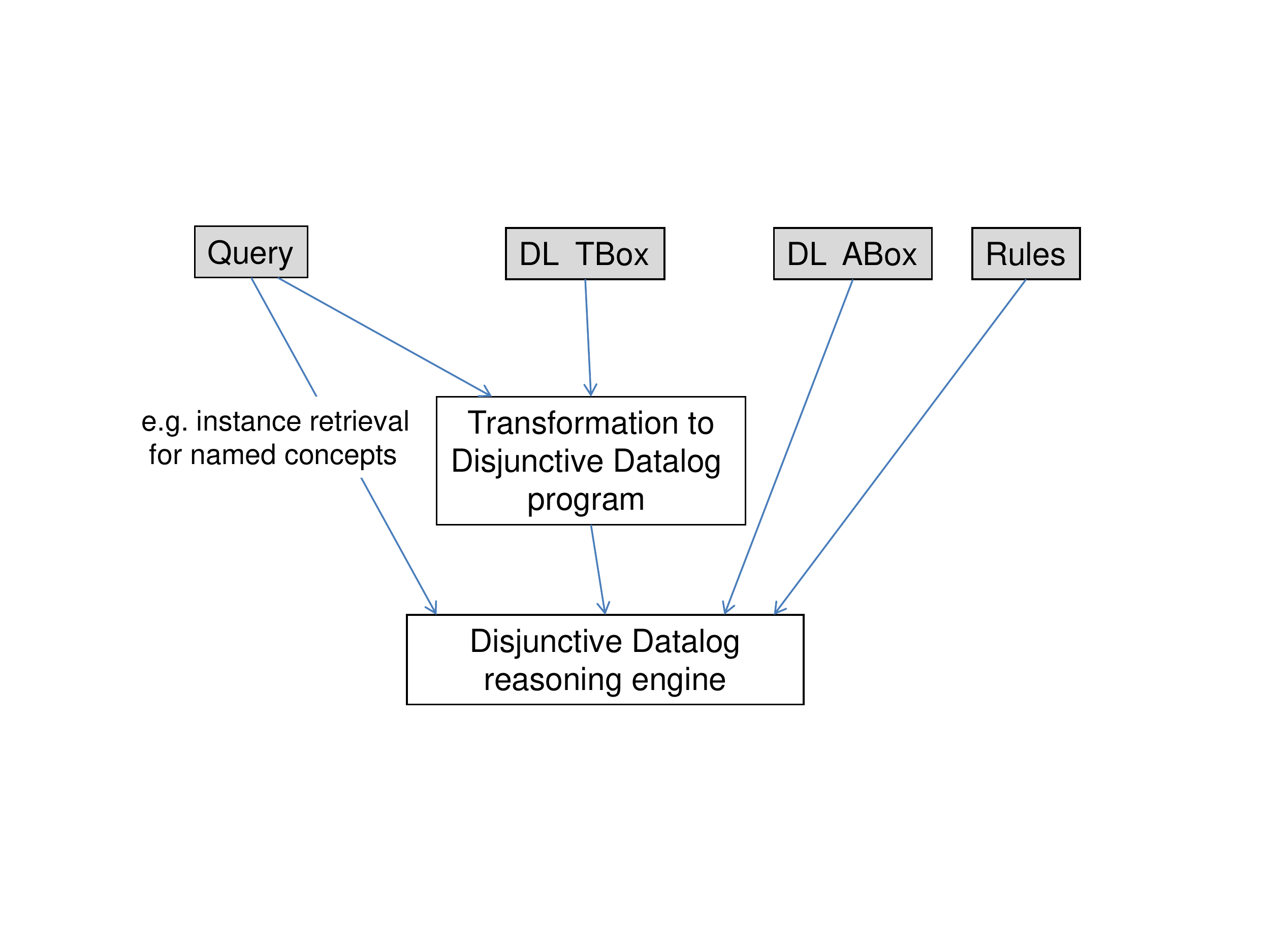}
\caption{Reasoning in \textsc{KAON2}.}
\label{fig:KAON2}
\end{figure}
\\
\indent
Figure \ref{fig:KAON2} presents an overview of the reasoning in \textsc{KAON2}.
\\
\indent
We implemented the proposed method for pattern mining in a system called \textsc{SEMINTEC}\footnote{http://www.cs.put.poznan.pl/alawrynowicz/semintec.htm} (\textit{Semantically-enabled data mining techniques}). Our implementation is written in Java (version 1.5). It uses \textsc{KAON2}'s\index{KAON2} API to manipulate and reason on combined knowledge bases. 
Figure \ref{fig:SEMINTEC_io} presents the input and output of our system and illustrates the interaction with the reasoner.
As an input to the system, the user is expected to provide the following files: \emph{setup file} (in XML format, with the parameters of the execution such as the logical and physical URI of the knowledge base, reference concept, minimum support threshold etc.) and \emph{knowledge base files} (in OWL and SWRL\footnote{www.w3.org/Submission/SWRL/} formats). As an output the system generates the files with: \emph{frequent patterns} discovered during the execution, \emph{statistics} of the execution, and a file with a \emph{trie} that stores patterns, in XML-based GraphML\footnote{http://graphml.graphdrawing.org} format. The implementation of \textsc{SEMINTEC} is publicly available\footnote{http://www.cs.put.poznan.pl/alawrynowicz/semintec.htm}. 

\begin{figure}[t]
\centering\includegraphics[width=0.8\textwidth]{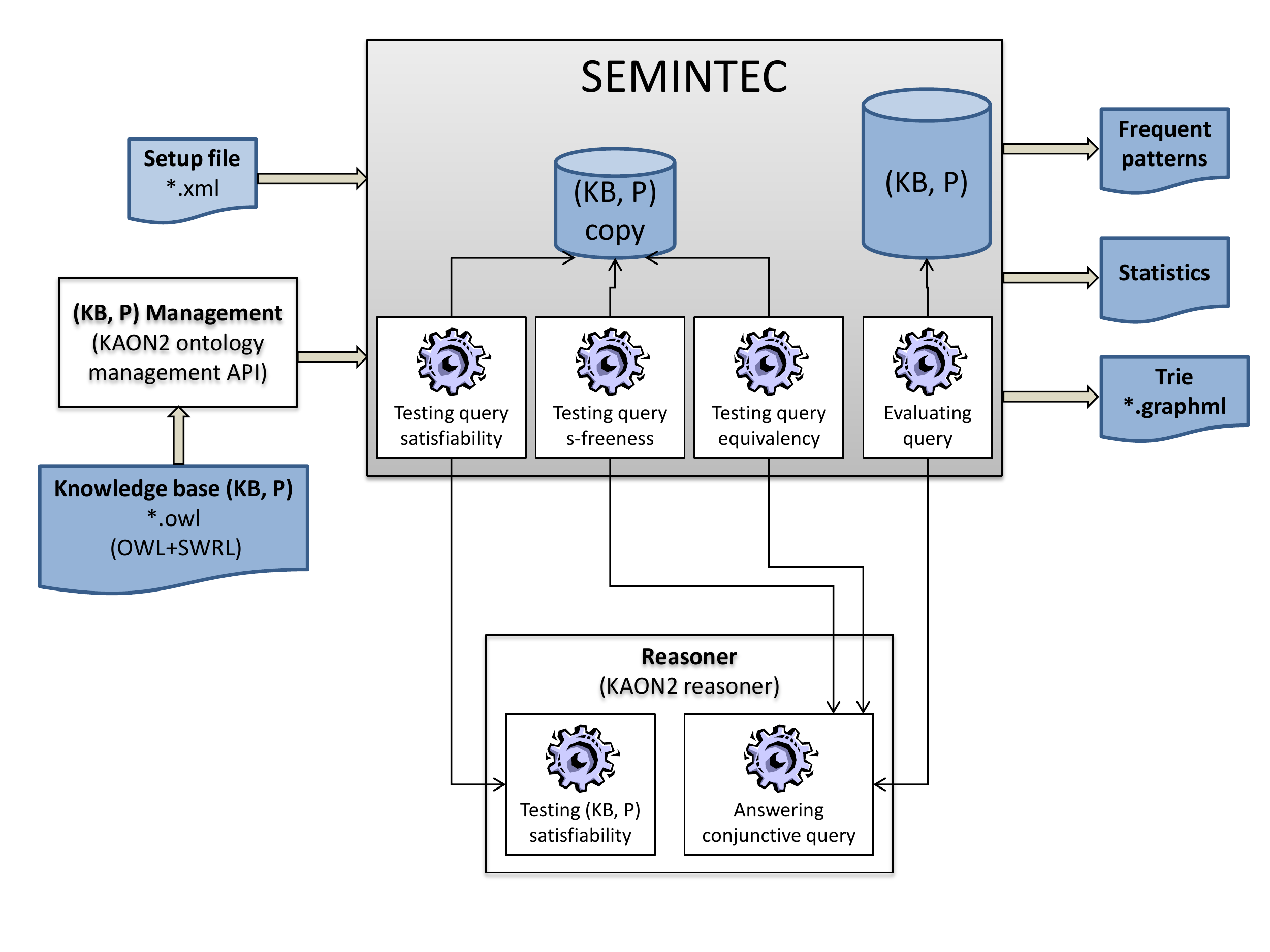}
\caption{\textsc{SEMINTEC} input/output and interaction with the reasoner.}
\label{fig:SEMINTEC_io}
\end{figure}

\section{Experimental evaluation}\label{sec:experiments}
In this section, we present an experimental evaluation of the proposed method for frequent pattern mining with the focus on the usefulness of exploiting the semantics of the knowledge base at different steps of our algorithm. In particular, the goals of the experiments were to investigate the influence of using intensional background knowledge expressed in DL with DL-safe rules on the data mining efficiency (i.e., computing time) and the quality of the results (i.e., the number and the form of the discovered patterns). We wanted to test how our method performs on datasets of different sizes and complexities, in order to obtain an idea what kinds of ontologies can be handled efficiently. In particular, the experiments were supposed to answer the following questions:
\begin{itemize}
\item how using the intensional part of the background knowledge for the semantic tests of generated patterns influences the execution time and the results of pattern discovery?
\item how the complexity of the intensional background knowledge, in particular the types of DL constructors and DL axioms, influences the execution time and the results of pattern discovery?
\item how exploiting concept and role taxonomies influences the execution time of pattern discovery?
\end{itemize}

\emph{Test datasets} 
For the tests we used three datasets, whose general characteristics is presented in Table \ref{tab:datasets}. The (\textsc{FINANCIAL})\footnote{\textsc{FINANCIAL},  http://www.cs.put.poznan.pl/alawrynowicz/financial.owl} dataset was created on the basis of a dataset from the PKDD'99 Discovery Challenge as a part of our research presented in this paper, and currently is the part of the benchmark suite of \textsc{KAON2}\index{KAON2}. 
\textsc{FINANCIAL} ontology describes the domain of banking. \textsc{FINANCIAL} dataset is relatively simple, as it does not use existential quantifiers or disjunctions. It contains, however, functional roles and disjointness constraints. Thus, it requires equality reasoning, which is difficult for deductive databases. 
\\
\indent
\textsc{SWRC} ontology, as used in our experiments, was published at the \emph{4th International EON Workshop (EON2006)}\footnote{http://km.aifb.uni-karlsruhe.de/ws/eon2006}. It was a part of the testbed\footnote{http://km.aifb.uni-karlsruhe.de/ws/eon2006/ontoeval.zip} used in the ontology evaluation session at the workshop. \textsc{SWRC} ontology (''\emph{Semantic Web for Research Communities}'') represents knowledge about researchers and research communities. Instance data, published at the EON website, describes the AIFB Institute of the University of Karlsruhe. The TBox of this ontology contains concept inclusion axioms, universal quantification, but no existential quantifiers, and no disjunctions, so it is simple. By r\textsc{SWRC} we denote our extension of this dataset by the rules presented in Table \ref{tab:datasets}.
\\
\indent
\textsc{LUBM} is a benchmark from the Lehigh University\footnote{LUBM,  http://swat.cse.lehigh.edu/projects/lubm/}, consisting of a university domain ontology and a generator of synthetic data. Existential quantifiers are used, but no disjunctions or number restrictions occur, hence the reduction algorithm of \textsc{KAON2} produces an equality-free Horn program, on which query answering can be performed deterministically. In the experiments we used r\textsc{LUBM}, an extension of \textsc{LUBM} ontology by two rules (presented in Table \ref{tab:datasets}) which was proposed by the authors of the DL-safe rules component of Pellet in \cite{Pellet-DL-safe:2006}. 
\begin{table}[t]
\caption{Characteristics of the test datasets.}\label{tab:datasets}
\footnotesize %
\begin{center}
\begin{tabular}{l | l | l | l | l | l}
\hline
dataset & $DL$ & \#concepts & \#obj. roles & \#rules & \#individuals \\
\hline
\textsc{FINANCIAL} & $\mathcal{ALCIF}$ & 60 & 16 & 0 & 17941 \\
r\textsc{SWRC} & $\mathcal{ALI}$(\textbf{D}) & 55 & 44 & 3 & 2156\\
r\textsc{LUBM} & $\mathcal{SHI}$(\textbf{D}) & 43 & 25 & 2 & 17174 \\
\hline
\multicolumn{6}{l}{Rules in r\textsc{SWRC}} \\
\hline
\multicolumn{6}{l}{$p\_knowsAboutTopic(x, z)\leftarrow Person(x), worksAtProject(x, y),isAbout(y, z)$}\\
\multicolumn{6}{l}{$p\_coAuthoredByFullProfessor(x)\leftarrow Article(x), author(x, y), FullProfessor(y)$}\\
\multicolumn{6}{l}{$finances(x,z)\leftarrow Organization(x), finances(x, y), Project(y), isAbout(y, z)$}\\
\hline
\multicolumn{6}{l}{Rules in r\textsc{LUBM}}\\
\hline
\multicolumn{6}{l}{$GraduateStudent(x)\leftarrow Person(x), takesCourse(x, y), GraduateCourse(y)$}\\
\multicolumn{6}{l}{$p\_specialCourse(z)\leftarrow FullProfessor(x), headOf(x, y), teacherOf(x, z)$}\\
\hline
\end{tabular}
\end{center}
\end{table}
\\\\
\emph{Test setting} 
All tests were performed on a PC with Intel Core2 Duo 2.4GHz processor, 2GB of RAM, running Microsoft Windows Server 2003 Standard Edition SP1. The JVM heap size was limited to 1.5GB. We used the version of \textsc{KAON2}\index{KAON2} released on 2008-01-14.   

\subsection{Results of the experiments}\label{sec:results}
\subsubsection{Analysis whether semantic tests of generated patterns are useful}\label{sec:ex_nosemresults} 
The goal of this experiment was to compare the setting where intensional background knowledge was used for testing generated candidates as well as for evaluating them with the setting where the background knowledge was used only during the candidate evaluation. We were interested in efficiency and quality of the results. 
The bias consisted of restricting the predicates, used to build patterns, only to those having any extension (to avoid testing predicates without any assertions), and giving new names to all variables in the newly added dependent atoms, except the variables shared with the last atom in a query.  
\\
\indent
In the first setting, \textsc{SEM}, the original algorithm for query expansion was used, that is Algorithm \ref{alg:expandNode}. In the second setting, \textsc{NOSEM}, the algorithm was run without the steps for checking pattern satisfiability, s-freeness and equivalence with already found frequent patterns, that is, lines 4-6 from Algorithm \ref{alg:expandNode} were omitted. However, the other parts of the solution such as the trie data structure as well as the techniques for reducing syntactic redundancy based on the trie were left unchanged, and used in the second setting as well. Hence, some assumptions made for the kinds of patterns expected as the result of the execution of our method were applied for both settings. In particular, syntactically non-redundant copies of atoms, in which output variables were given new names, were not generated as dependent atoms in both settings. Not generating copies of atoms, which is based on the assumption of generating only s-free candidate patterns, greatly influences the time and the results of the pattern mining, as without the semantic tests for redundancy, one could not avoid chains like: $Client(x), isOwnerOf(x, y_1), isOwnerOf(x, y_2), isOwnerOf(x, y_3),...$. Thus, we compare our proposed setting with the one which is not strictly naive and which lacks the most time consuming operations.
\\
\indent
The parameters measured during an execution of the experiment, were: (i) running time ($runtime$), (ii) number of candidate patterns ($cand$), (iii) number of frequent patterns ($freq$).
Good results are characterized by low number of candidates and frequent patterns, and short running time. Additionally, a ratio of frequent patterns to candidate patterns should be as high as possible, that is, as few as possible unproductive candidate patterns should be evaluated.  
\\\\
\emph{Qualitative analysis} 
Below we present and discuss some patterns discovered during the experimental evaluation. We restrict the analysis to the ontologies with real (nonsynthetic) data. 
\\
\indent
The following is one of the longest patterns discovered from the \textsc{FINANCIAL} dataset, by our method (\emph{SEM} setting):\\\\
$Q_{SEM1}(key)=Client(key),hasOwner(x_1, key),hasStatementIssuanceFrequency(x_1, x_2), \\Monthly(x_2),hasPermanentOrder(x_1, x_3), isPermanentOrderFor(x_3, x_5), Household-\linebreak Payment(x_5), hasAgeValue(key, x_7),hasSexValue(key, x_8),FemaleSex(x_8),livesIn(key, x_{10});\\ support=$0.29
\\\\
It describes "\emph{a client who is an owner of an account with monthly statement issuance frequency, and with a permament order for household payment, who is a female, lives in some region, and is at some age}". The information that $Account$ is here the domain of $hasOwner$ and $Region$ is the range of $livesIn$ comes from the \textsc{FINANCIAL} ontology. One may notice, that the region in which the client lives and the age at which she is, is not specified in this pattern. Example, shorter patterns discovered, that involve roles $hasAgeValue$ or $livesIn$ and precise their range are shown below:\\\\
$Q_{SEM2}(key)=Client(key), hasOwner(x_1, key), hasStatementIssuanceFrequency(x_1, x_2),\\  
Monthly(x_2), hasAgeValue(key, x_4), From35To50(x_4); support$=0.21\\
$
Q_{SEM3}(key)=Client(key), livesIn(key, x_1), NorthMoravia(x_1); support$=0.17
\\\\
An example of a pattern discovered by running \emph{NOSEM} setting is as follows:
\\\\
$
Q_{NOSEM1}(key)=Client(key), livesIn(key, x_1), Region(x_1); support$=1.0
\\\\
The pattern $Q_{NOSEM1}$ has the semantically redundant atom, $Region(x_1)$, due to the specification of $Region$ as the range of role $livesIn$ in the \textsc{FINANCIAL} $KB$. 
\\
\indent
Let us now present the example patterns discovered from the r\textsc{SWRC} dataset. By running the \emph{SEM} setting, the following example patterns have been discovered:\\\\
$
Q_{SEM4}(key)=Person(key), author(x_1, key), publication(x_2, x_1),  
p\_knowsAboutTo-\linebreak pic(x_2, x_3);support$=0.70
 \\
$
Q_{SEM5}(key)=Person(key), author(x_1, key), publication(key, x_2), 
Publication(x_2); \\ support$=0.75
\\\\
The meaning of pattern $Q_{SEM4}$ may seem unclear with regard to the r\textsc{SWRC} knowledge base. In the knowledge base neither ranges nor domains of $author$ and $publication$ are specified. However, from the rule defining $p\_knowsAboutTopic$ we know that its first argument represents $Person$ and the second one $Topic$. Thus, we may conclude that the pattern says that ''\emph{some person, who knows about some topic, is related to the publication who is authored by the person represented by the reference concept}''. From the intensional part of the $(KB, P)$ we do not know about the nature of this relation for $Person$ as role $publication$ is missing domain and range specifications. By deeper analysis of the knowledge base, we may notice that concept $AcademicStaff$, that is the subconcept of $Person$, is subsumed by concept $\forall publication$.$Publication$. Thus for academic staff, a particular type of persons, $publication$ range is $Publication$.
\\
\indent  
In pattern $Q_{SEM5}$, a person is related by role $publication$ with some $Publication$. By the common sense reasoning, this pattern carries redundant information. It is, however, s-free, as the range of $publication$ is not specified in the $(KB, P)$.
\\
\indent
With regard to the \emph{NOSEM} setting let us discuss the following pattern:
\\\\
$
Q_{NOSEM2}(key)=Person(key), publication(key, x_1), Publication(x_1), InProceedings(x_1);\\ support$=0.47
\\\\
Since in the $(KB, P)$, $InProceedings$ is the subconcept of $Publication$,  atom $Publication(x_1)$ is semantically redundant.
\\\\
\emph{Quantitative analysis} 
Table \ref{tab:results_nosem} shows the results for a selected support threshold for each dataset. The results are shown up to the lengths of patterns where either an execution of the proposed method (\emph{SEM}) has not exceeded the threshold of 24 hours of the running time (r\textsc{SWRC}, r\textsc{LUBM}) or the whole trie was generated in this setting (\textsc{FINANCIAL}).  
\begin{table}[t]
\caption{Results of the experiment on effectiveness of the semantic tests.}
\label{tab:results_nosem}
\footnotesize
\begin{center}
\begin{tabular}{r | r | r | r | r | r | r | r | r | r }
\hline
\multicolumn{1}{c|}{Max} & \multicolumn{4}{c|}{number of patterns} & \multicolumn{2}{c|}{reduction} & \multicolumn{2}{c|}{runtime[s]} & 
\multicolumn{1}{c}{speedup}\\
\cline{2-10}
Length & \multicolumn{2}{c}{NOSEM} &	\multicolumn{2}{c|}{SEM} & \multicolumn{2}{c|}{NOSEM/SEM} &	\multicolumn{1}{c}{NOSEM} &	\multicolumn{1}{c|}{SEM} & NOSEM/SEM \\
     & \multicolumn{1}{r}{cand} & \multicolumn{1}{r}{freq} & \multicolumn{1}{r}{cand} & freq & \multicolumn{1}{r}{cand} & freq & \multicolumn{1}{r}{} & & \\
\hline
\multicolumn{8}{l}{\textsc{FINANCIAL}
, $minsup$=0.2, reference concept=$Client$}\\
\hline
1	&		1	&	1	&	1	&	1	&	1.00	& 1.00	&	0.5	&	0.5	&	\textbf{1.07} \\
2	&		91	&	9	&	15	&	7	&	\textbf{6.07}	& \textbf{1.29}	&	42.7	&	14.2	&	\textbf{3.01} \\
3	&	  582 &	69	&	71	&	27	&	\textbf{8.20}	& \textbf{2.56}	&	303.5	&	104.5	&	\textbf{2.91} \\
4	&	 2786 & 479 &	253 &	68	&	\textbf{11.01}	& \textbf{7.04}	&	2931.6 &	569.7	&	\textbf{5.15} \\
5	&	 -	   &	-	&	569 &	131 &		-	&	-		&	-		 &  2166.9 &	-\\
6	&	-		&	-	&	1009 & 214 &	-	&	-		&	-		 &  6042.7	& -\\
7	&	-		&	-	&	1524 & 303 &	-	&	-		&	-		 &  12204.8	& -\\
8	&	-	   &	-	&	1963 & 376 &	-	&	-		&	-		 &  20200.0	& -\\
9	&	-		&	-	&	2307 & 421 &	-	&	-		&	-		 &  26346.1	& -\\
10	&	-		&	-	&	2513 & 440 &	-	&	-		&	-		 &  29614.2	& -\\
11	&	-		&	-	&	2608 & 444 &	-	&	-		&	-		 &  30309.6	& -\\
12	&	-		&	-	&	2634	& 444& 	-	&	-		&	-		&	30821.6	& -\\
\hline
\multicolumn{8}{l}{r\textsc{SWRC}, $minsup$=0.3, reference concept=$Person$}\\
\hline
1 &	   1 &   1 &	1	 &    1	& 1.00	& 1.00	& 0.1	& 0.1	 & \textbf{1.01} \\
2 &     92 &   3 &  92	 &    3	& 1.00	& 1.00	& 7.4	 & 16.5	& 0.45 \\
3 &    279 &  22 & 271	 &   14	&  \textbf{1.03}	& \textbf{1.57}	& 23.0	& 155.2	& 0.15 \\
4 &   1556 & 272 & 913	 &  100	&  \textbf{1.70}	& \textbf{2.72}	& 169.1	& 2533.5	& 0.07 \\
\hline
\multicolumn{8}{l}{r\textsc{LUBM}, $minsup$=0.3, reference concept=$Person$}\\
\hline
1 &	1	& 1    &	1	   & 1	& 1.00	& 1.00	& 0.3	   & 0.3	   & 1.00 \\
2	& 68	& 7    & 67	   & 6	& \textbf{1.01}	& \textbf{1.17}	& 12.5	& 16.5	& 0.76 \\
3	& 361	& 63   &	269	& 31	& \textbf{1.34}	& \textbf{2.03}	& 82.8	& 142.6	& 0.58 \\
4	& 2885 &	789 &	1438	& 194	& \textbf{2.01}	& \textbf{4.07}	& 9713.0	& 3486.7	& \textbf{2.79} \\
\hline
\end{tabular}
\end{center}
\end{table}
From the presented results one can conclude that with regard to the reduction in the number of patterns, there is a gain for all datasets, reaching 11.01 times for candidate patterns and 7.04 times for frequent patterns in case of the \textsc{FINANCIAL} dataset. 
\\
\indent
With regard to the running time, in case of \textsc{FINANCIAL} dataset, the speedup has been reached for all maximum lengths of patterns. For longer patterns, the \emph{NOSEM} setting was unable to finish execution in 24 hours, while executing the \emph{SEM} setting allowed to generate the whole trie of frequent patterns for \textsc{FINANCIAL} dataset.  In case of r\textsc{LUBM} dataset, the speedup has been reached for the longest, most important, maximum pattern length. For r\textsc{SWRC}, however, the \emph{NOSEM} setting was significantly better with regard to the running time.
\\
\indent
We also measured the method performance for different minimum support thresholds. The results are reported in Figure \ref{fig:results_nosem}. The bars representing the numbers of frequent patterns are superimposed on those, representing the numbers of candidate patterns. In case of the \textsc{FINANCIAL} dataset the differences between the numbers of candidate patterns in the \emph{SEM} setting in comparison to the \emph{NOSEM} setting are the largest from among those of the tested datasets. The number of candidate patterns in the \emph{SEM} setting constitutes about 9\% of that in the \emph{NOSEM} setting.  For the r\textsc{LUBM} dataset this ratio is about 45\% on average and for r\textsc{SWRC} is about 59\% on average. The number of frequent patterns in the \emph{SEM} setting is on average equal to about 14\% of the number of frequent patterns in the \emph{NOSEM} setting for the \textsc{FINANCIAL} dataset, about 22\% for r\textsc{LUBM} dataset and about 37\% for r\textsc{SWRC} dataset. Since in case of the \textsc{FINANCIAL} dataset, the differences in pattern numbers between the \emph{SEM} setting and the \emph{NOSEM} setting are the largest from among the tested datasets, relatively the biggest number of semantically redundant patterns is pruned away for this dataset, while for the r\textsc{SWRC} dataset this number is the lowest one.
\\
\indent
Let us now discuss the ratio between the number of frequent and the number of candidate patterns in case of the \emph{SEM} setting.  For the \textsc{FINANCIAL} dataset this ratio is equal 26\% on average, for the r\textsc{LUBM} dataset 13\% on average, and for the r\textsc{SWRC} 11\% on average. Thus, in case of the \textsc{FINANCIAL} dataset relatively the least computation is done to evaluate useless candidate patterns. In case of the r\textsc{SWRC} dataset the computational effort is relatively the largest. 

Summarizing, the semantic tests performed during the pattern generation were useful in terms of the number of patterns for all datasets, and in the running time for \textsc{FINANCIAL} and r\textsc{LUBM} datasets but not for r\textsc{SWRC} dataset. They were most useful for the \textsc{FINANCIAL} dataset, where relatively the least number of patterns were generated and tested in the \emph{SEM} setting in comparison to the \emph{NOSEM} setting, and where the ratio between frequent and candidate patterns in the \emph{SEM} setting was the biggest. The semantic tests were least useful in case of the r\textsc{SWRC} dataset.

\begin{figure}
  \centering
  \subfloat[Number of patterns]{\includegraphics[width=0.5\textwidth]{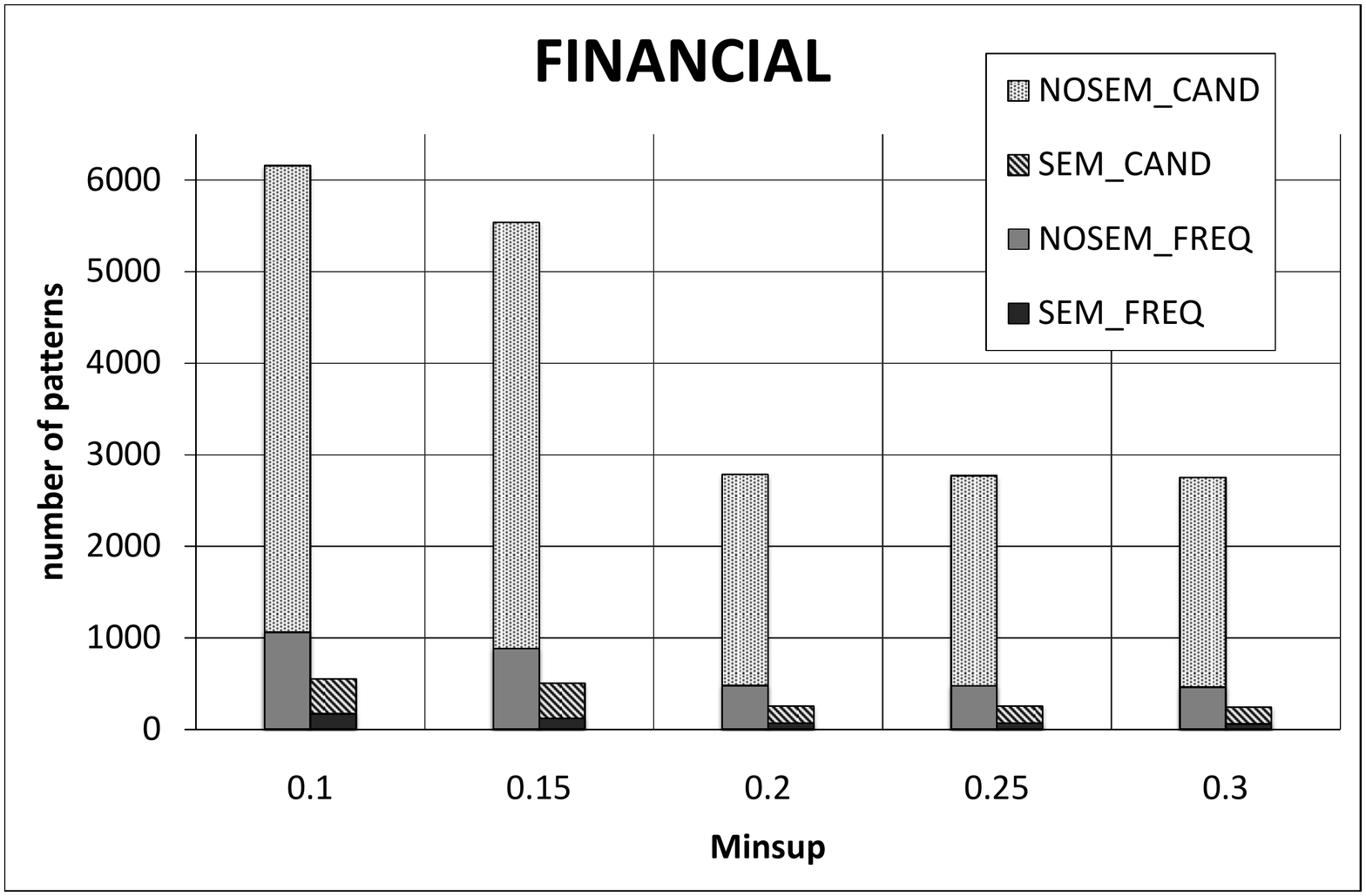}
}                
\subfloat[Running time]{\includegraphics[width=0.5\textwidth]{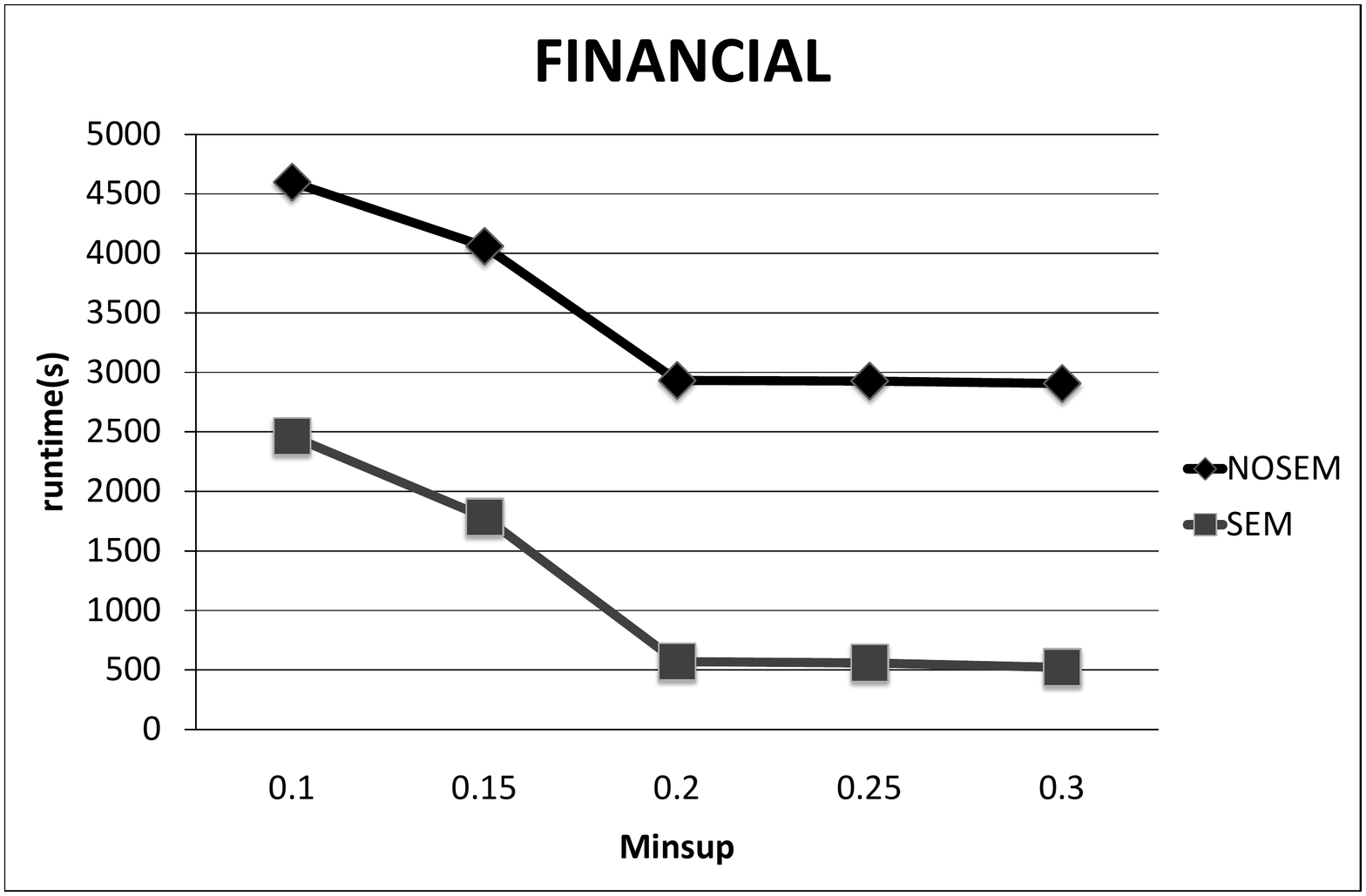}
}                
\\
  \subfloat[Number of patterns]{\includegraphics[width=0.5\textwidth]{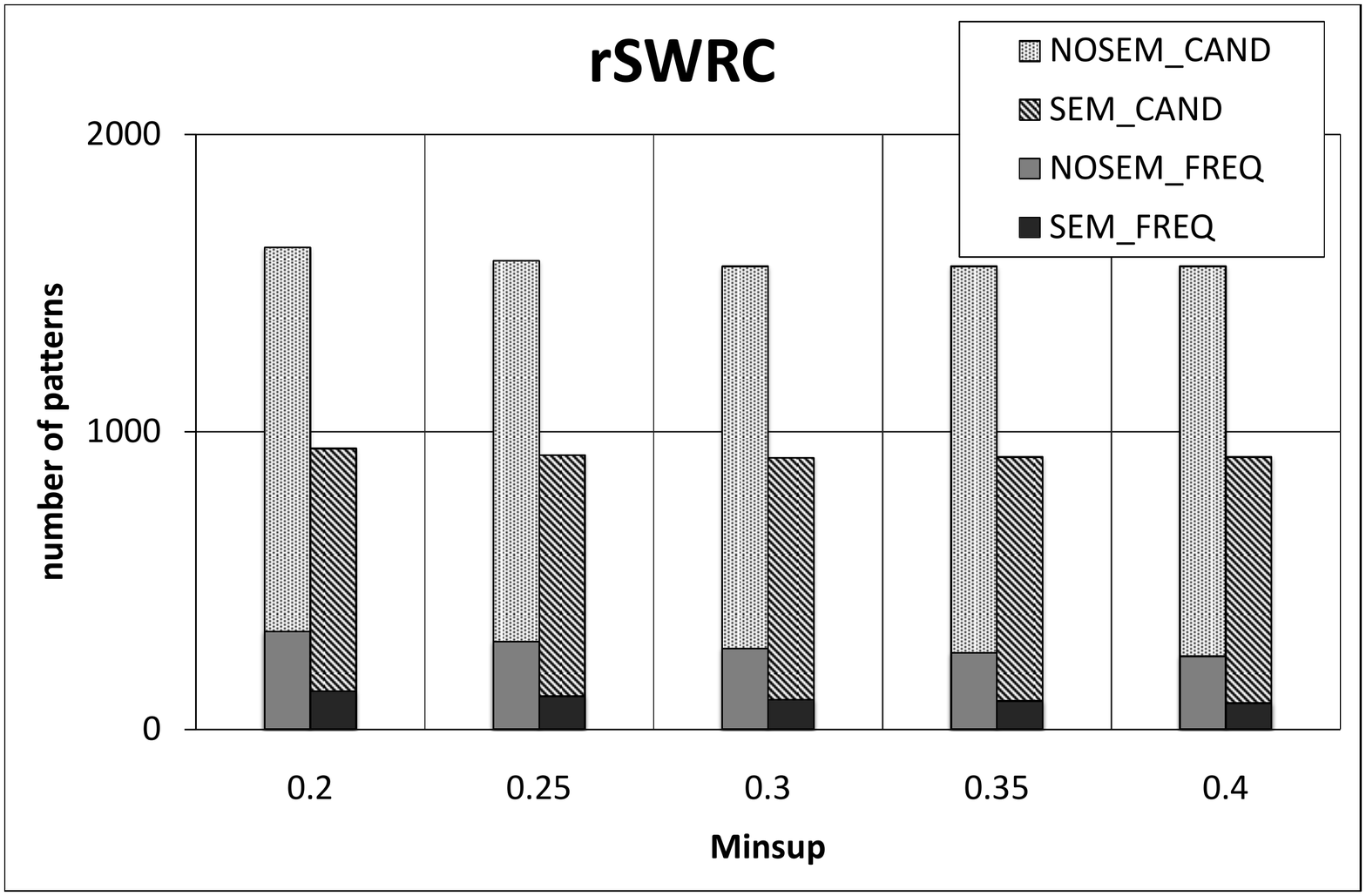}
}                
\subfloat[Running time]{\includegraphics[width=0.5\textwidth]{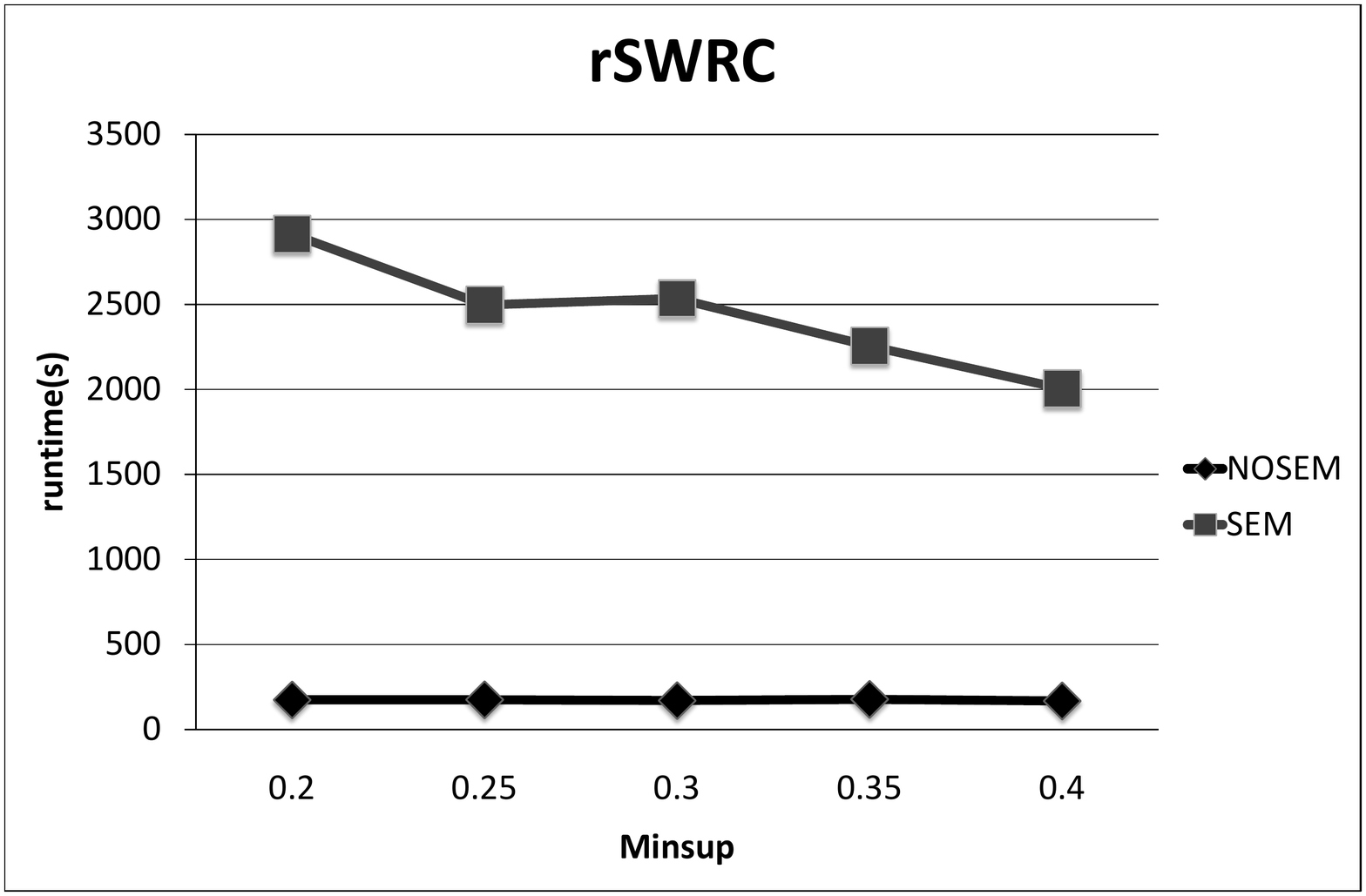}
}                
\\
  \subfloat[Number of patterns]{\includegraphics[width=0.5\textwidth]{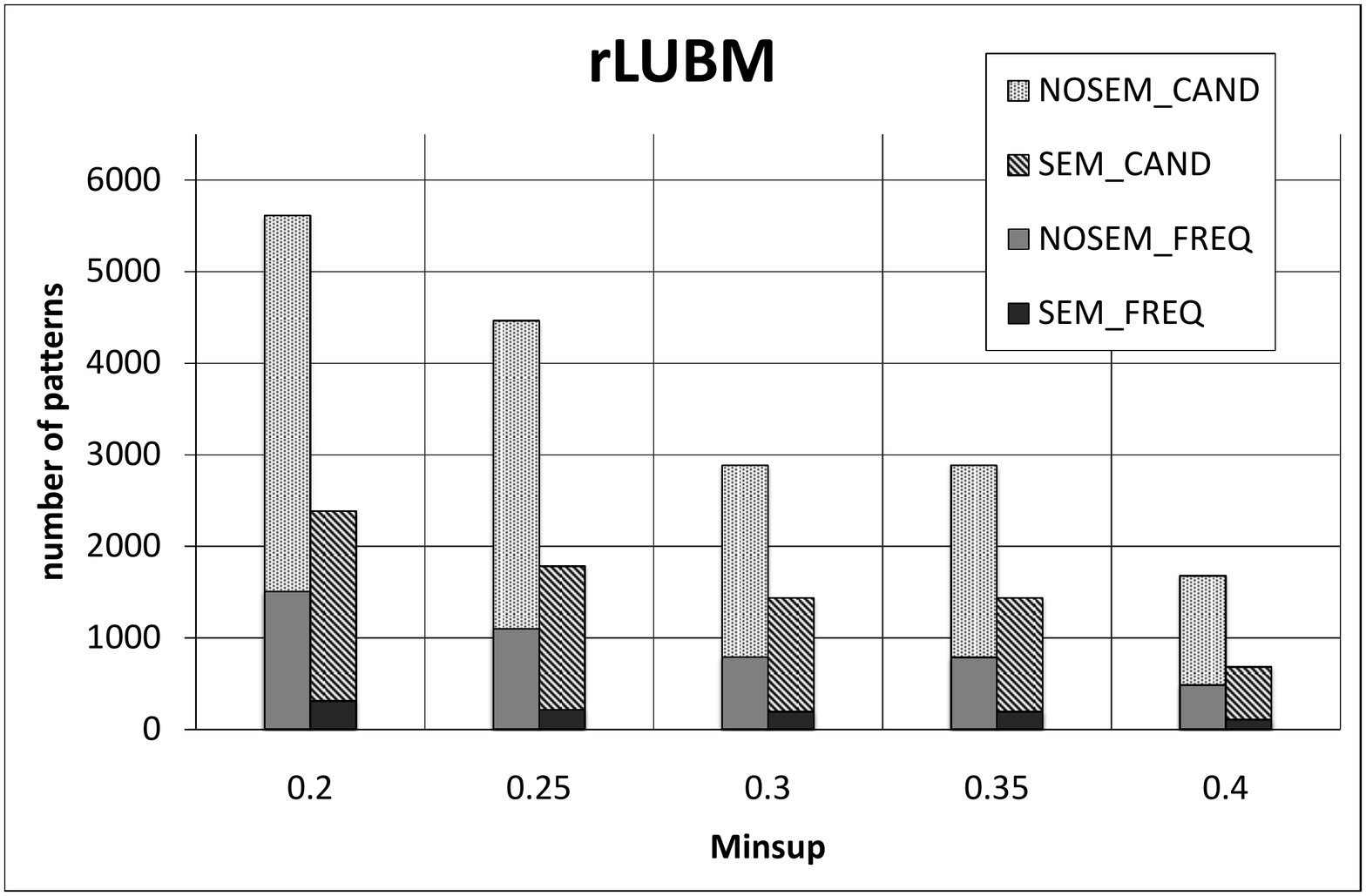}
}                
\subfloat[Running time]{\includegraphics[width=0.5\textwidth]{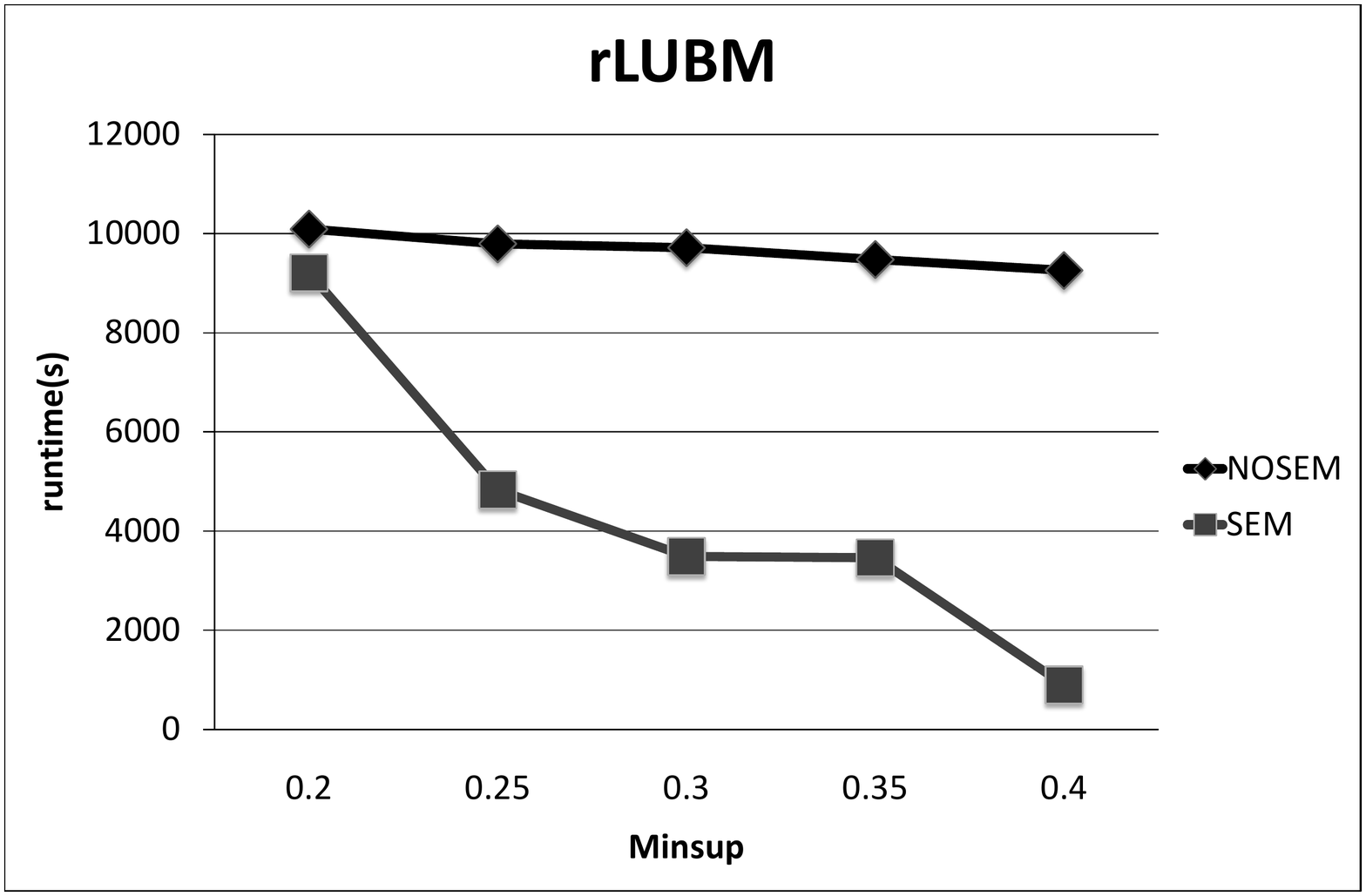}
}                
\caption{Results of the experiment, MAXLENGTH=4, \textsc{FINANCIAL}: $\hat{C}$=$Client$, r\textsc{SWRC}: $\hat{C}$=$Person$, r\textsc{LUBM}: $\hat{C}$=$Person$.}
  \label{fig:results_nosem}
\end{figure}

After the analysis of the results for r\textsc{SWRC} dataset one may pose the following question: \emph{should the semantic tests on patterns be performed together with checking their frequency or they should be performed afterwards as a postprocessing step of pattern mining?} 
For the \textsc{FINANCIAL} and r\textsc{LUBM} datasets it is clear that it was better to perform the tests together with pattern evaluation. The running times in \emph{SEM} setting (at least for the longest patterns in case of r\textsc{LUBM}) are already shorter than those in \emph{NOSEM} setting. For the r\textsc{SWRC} dataset we performed additional test. We took the patterns generated as the output of \emph{NOSEM} setting (MAXLENGTH = 4) and postprocessed them, leaving at the first step only s-free ones, and at the second step only one representative of each equivalence class of patterns. The additional execution time was 907.5s, which together with the execution time of \emph{NOSEM} setting, 169.1s (as specified in Table \ref{tab:results_nosem}), gives 1076.6s. This time is shorter than the time of the \emph{SEM} setting execution which is 2533.5s. That is, in case of r\emph{SWRC} dataset it was faster to perform data mining without semantic tests at the first step and then perform the tests as a postprocessing step.
\\
\indent
Why r\textsc{SWRC} dataset is especially hard for our approach, while the others are not, is discussed in Section \ref{sec:ex_expresivity}, which provides more insight into the influence of using intensional background knowledge during pattern generation. It is also noteworthy, that for the same intensional background knowledge, but for a bigger number of assertions (more probable case for data mining applications) it may be better to perform the semantic tests together with pattern evaluation. As an empirical proof of this claim we provide the experimental results in the following paragraph.
\\\\
\emph{Influence of the size of the dataset on the effectiveness} 
We measured how our method scales in terms of the running time with growing instance data. For this reason we used the replication of the axioms in the assertional part of the knowledge base for \textsc{FINANCIAL} and r\textsc{SWRC}. Each assertional part of \textsc{FINANCIAL}$\_n$ and r\textsc{SWRC}$\_n$ was obtained by replicating the original assertional part $n$ times.  
For r\textsc{LUBM}, we used the results of the execution of the generator of synthetic data (downloaded from \textsc{KAON2}'s testbed). Each assertional part of r\textsc{LUBM}$\_n$ was generated automatically for the number $n$ of universities. 
\\
\indent
\begin{figure}
  \centering
\subfloat[$\hat{C}$=$Client$, $minsup$=0.2]{
\includegraphics[width=0.45\textwidth]{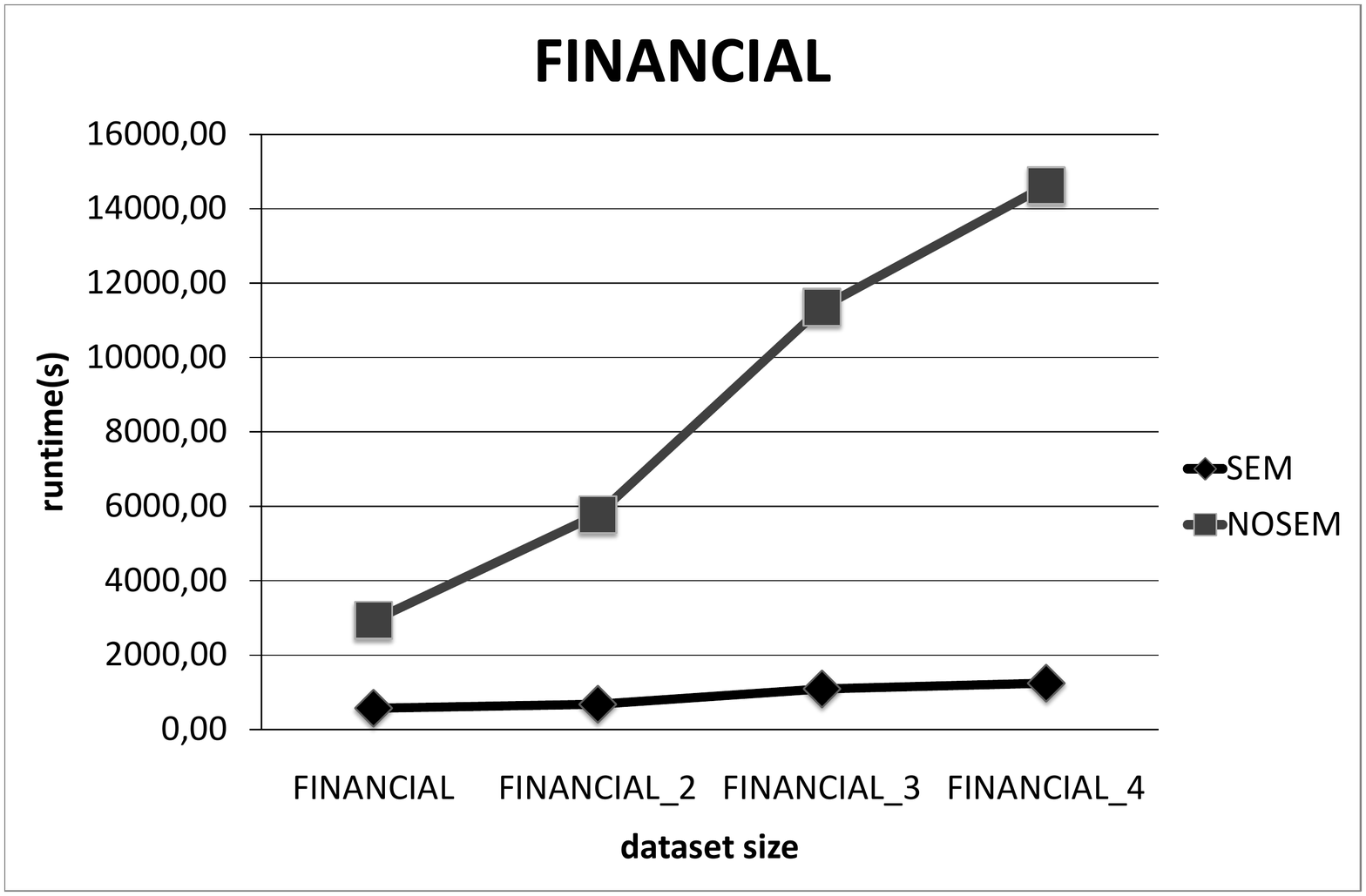}                
}
\subfloat[$\hat{C}$=$Person$, $minsup$=0.3]{
\includegraphics[width=0.45\textwidth]{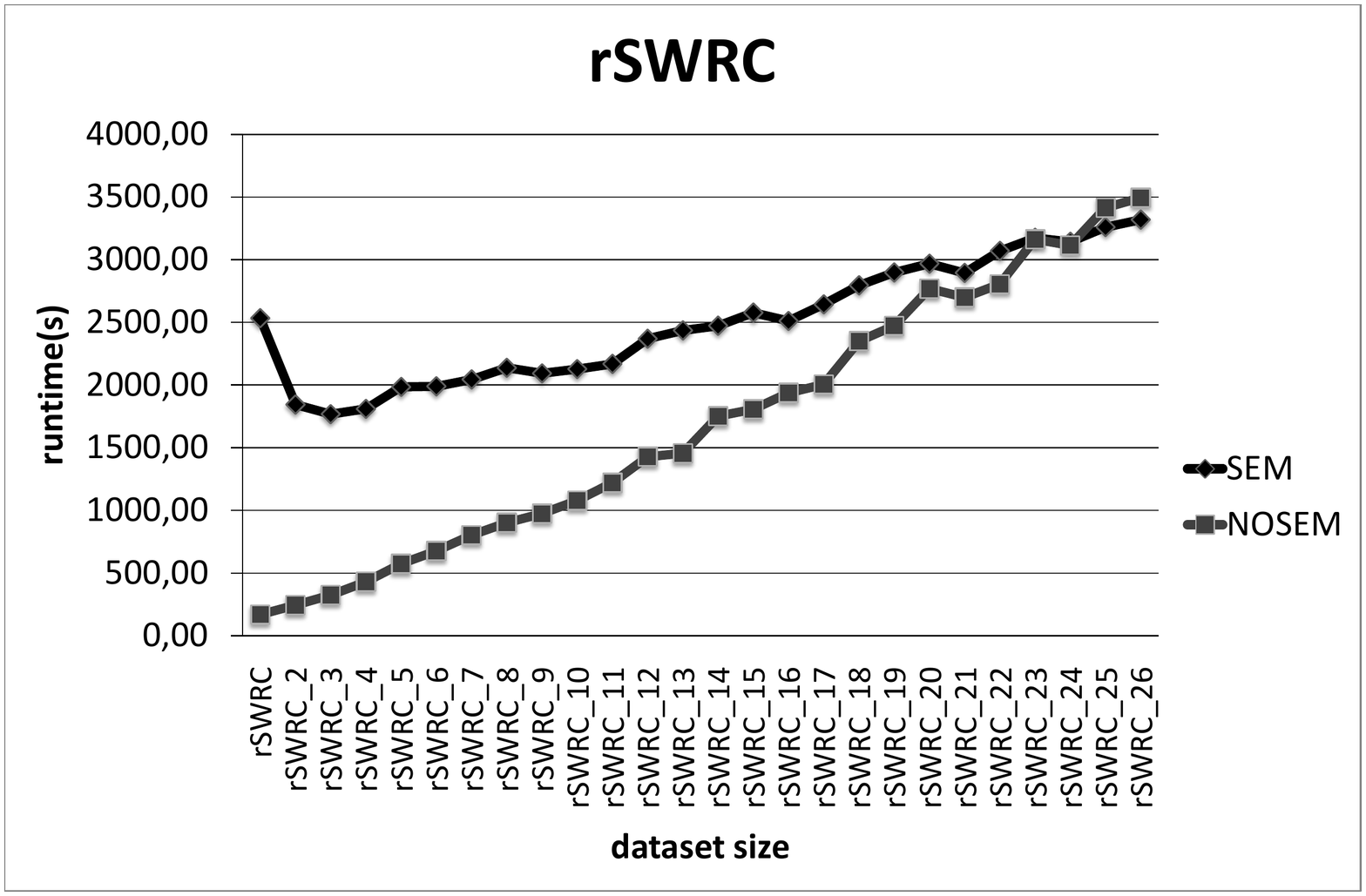}                
}
\\
\subfloat[$\hat{C}$=$Person$, $minsup$=0.3]{
\includegraphics[width=0.45\textwidth]{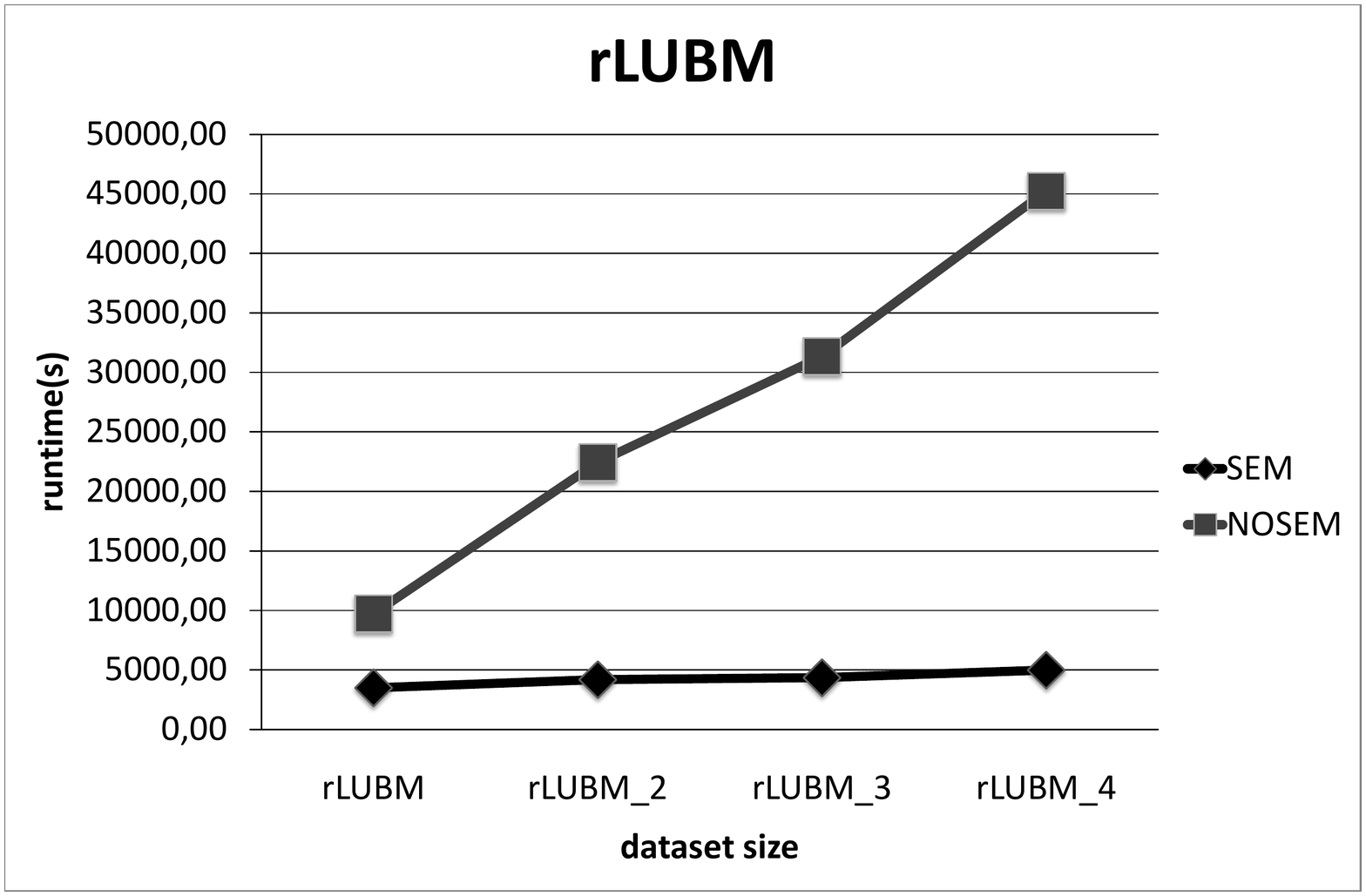}                
}
\caption{Results of the experiment, MAXLENGTH=4.}
\label{fig:results_aboxmulti_nosem}
\end{figure}
In the experimental results (Figure \ref{fig:results_aboxmulti_nosem}) one can observe that for all datasets, the bigger extensional part of the background knowledge, the relatively better performance of our method, \emph{SEM}, compared to the method, where no semantic tests are perfomed during the pattern generation, \emph{NOSEM}. That is, the overhead needed to compute the semantic tests becomes relatively smaller in comparison with the time needed to evaluate more queries on bigger sets of data. Especially interesting for us are the results on the problematic r\textsc{SWRC} dataset. One can observe that together with the growth of the assertional part of the background knowledge, the time needed to compute the \emph{NOSEM} setting increases more then the time needed to compute the \emph{SEM} setting. That's why, we perfomed the tests to show that for bigger volumes of data, described by the same intensional knowledge, a time needed to compute a trie of patterns in \emph{NOSEM} setting finally reaches and exceeds the time needed to compute the trie in the \emph{SEM} setting.   
\subsubsection{Influence of the expressivity of the dataset on the effectiveness}\label{sec:ex_expresivity}
In this experiment, additionally to the parameters measured in the experiment presented in Section \ref{sec:ex_nosemresults}, we collected the following information: (i) the number of candidates generated by the syntactic refinement rules ($gen$), (ii) the number of satisfiable candidate patterns ($sat$), (iii) the number of semantically free candidate patterns ($sfree$).
The goal was to investigate more deeply, how useful are the semantic tests on different types of datasets.
Figure \ref{fig:results_nosem_dl} shows the experimental results. 
\begin{figure}
  \centering
  \subfloat[$\hat{C}$=$Client$, $minsup$=0.2]{
\includegraphics[width=0.5\textwidth]{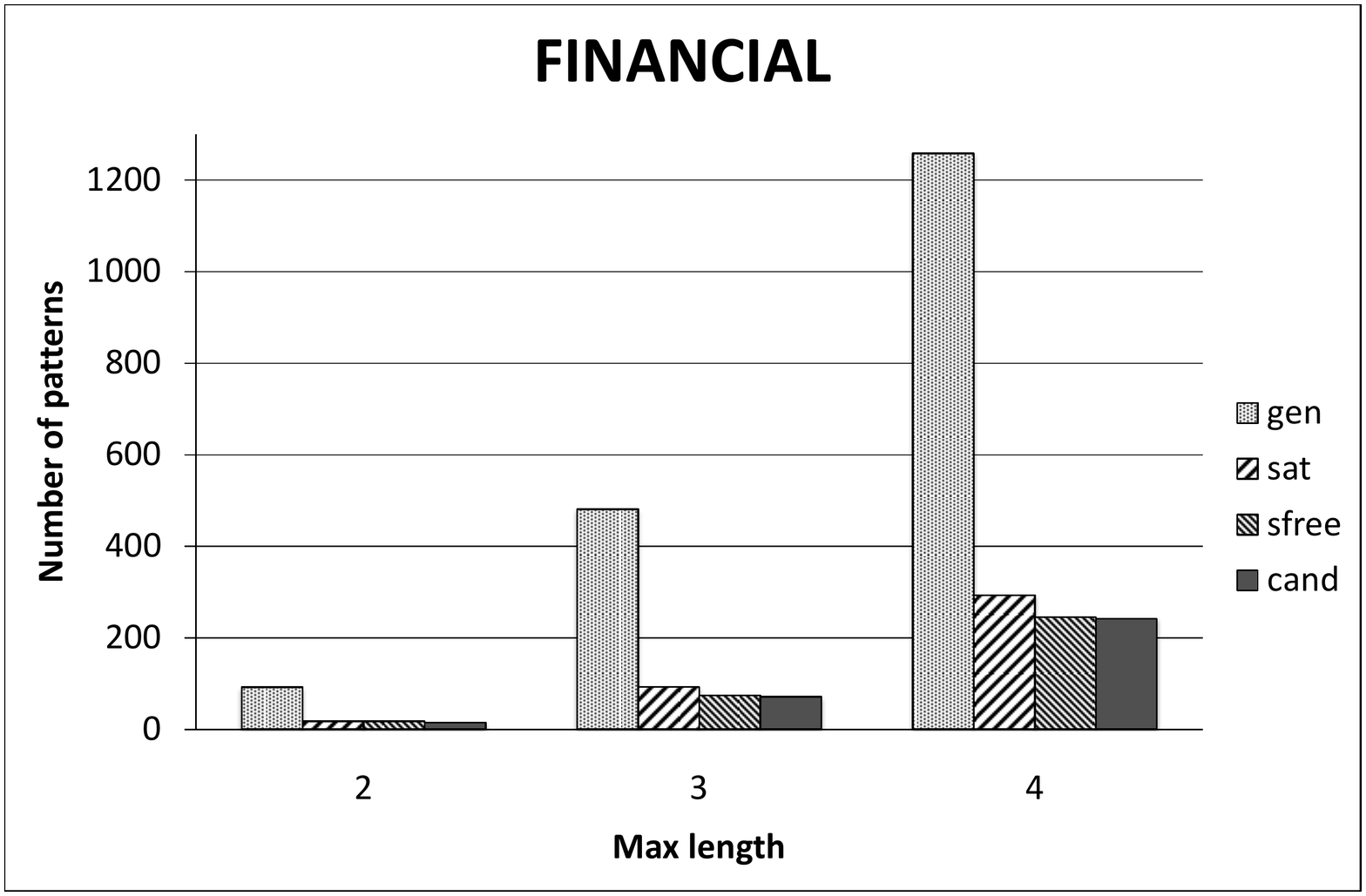} 
}                
\subfloat[$\hat{C}$=$Person$, $minsup$=0.3]{
\includegraphics[width=0.5\textwidth]{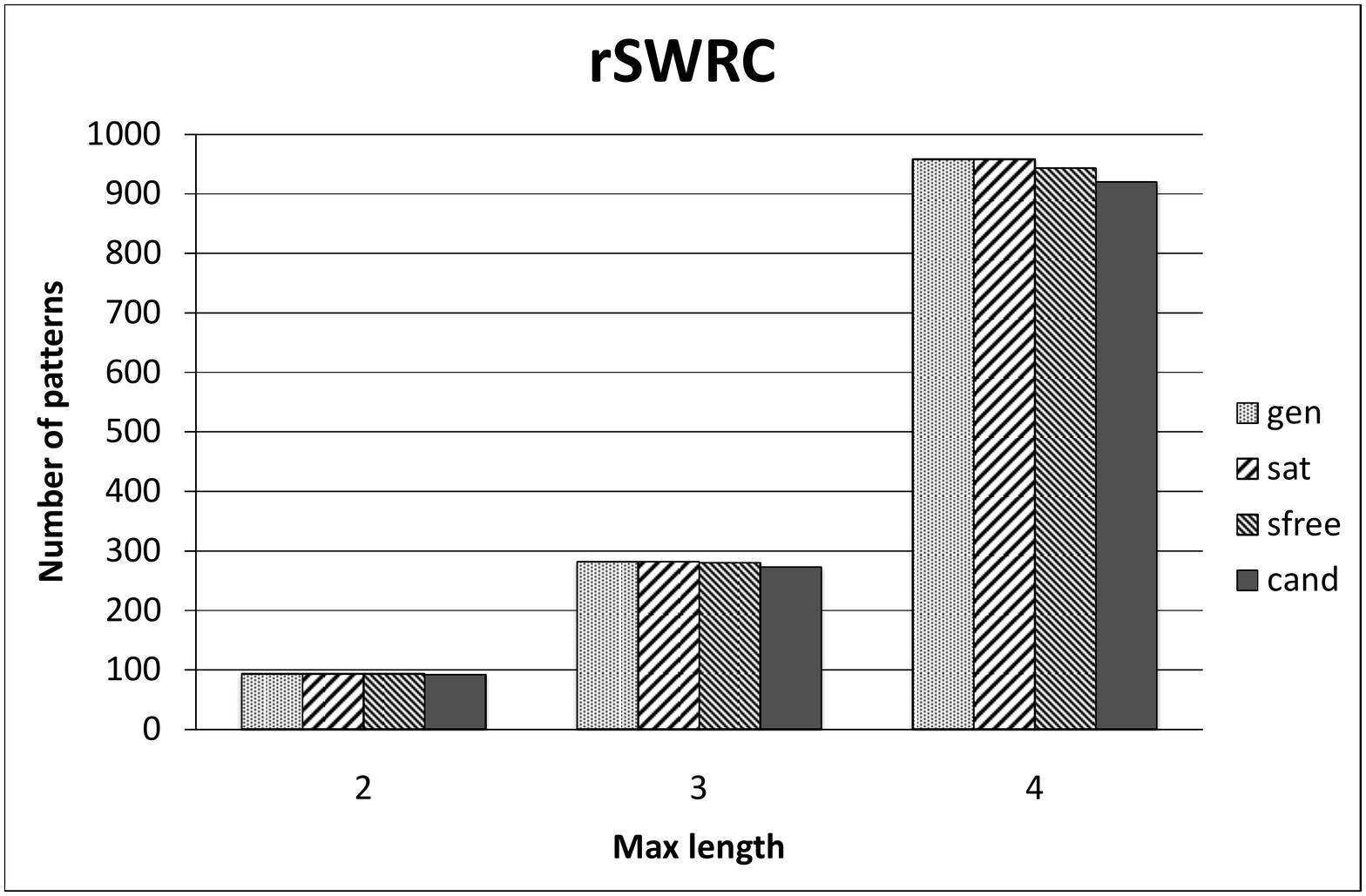} 
}  
\\              
\subfloat[$\hat{C}$=$Person$, $minsup$=0.3]{
\includegraphics[width=0.5\textwidth]{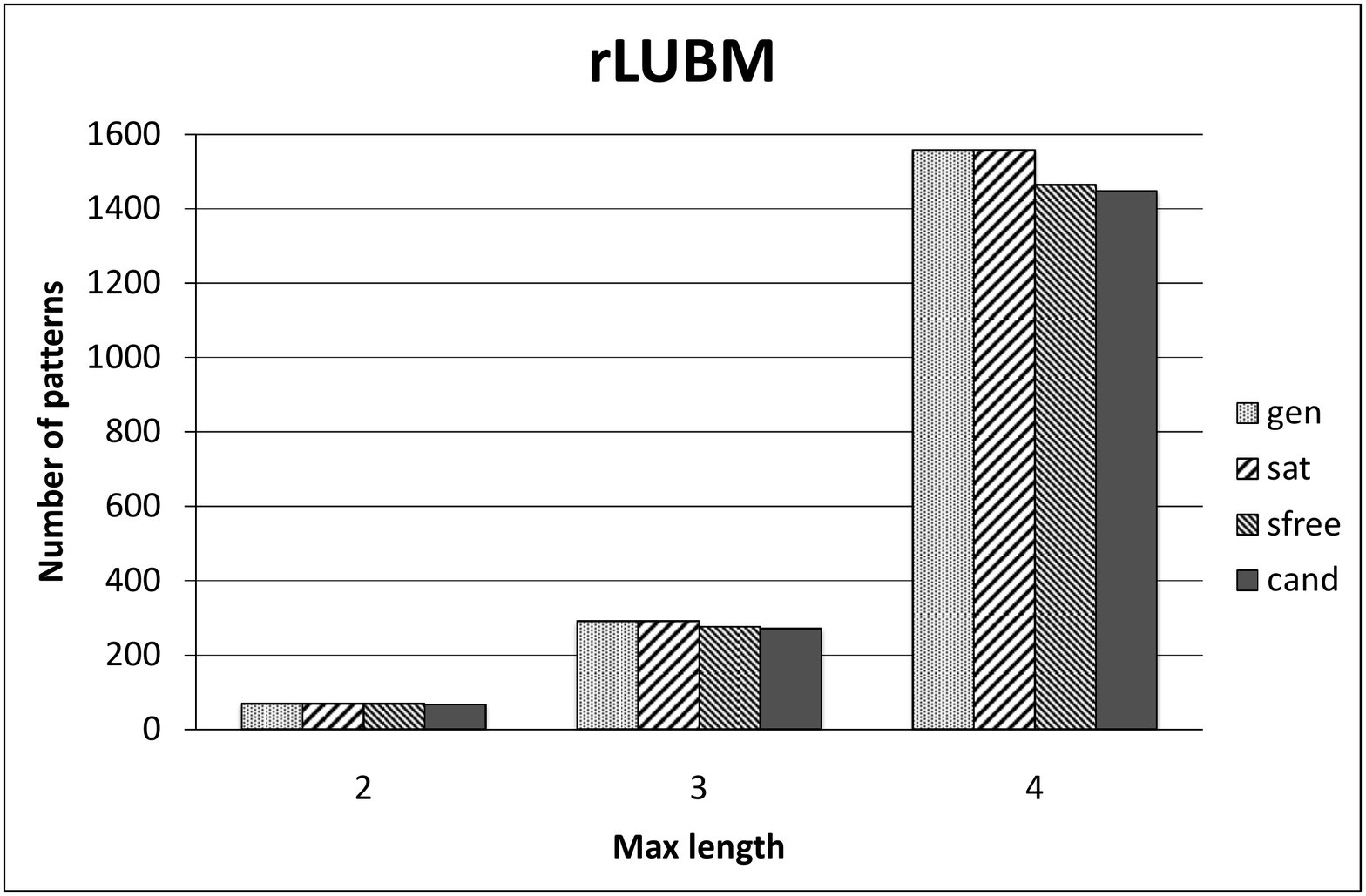} 
}                
\caption{Results of the experiment on the influence of the knowledge base expressivity on effectiveness, MAXLENGTH=4.}
  \label{fig:results_nosem_dl}
\end{figure}



The most important remark to be made after the analysis of the results is that after the \emph{satisfiability} test, in case of the \textsc{FINANCIAL} dataset many, while for the other two datasets none of the patterns were pruned away. We conclude that it is due to the \emph{disjointness constraints} present only in the \textsc{FINANCIAL} dataset, but not in the other ones. If two concepts are defined to be disjoint then adding an atom where some variable is described by one of them, while the other one describing the same variable is already present in a query is useless. For example, as concepts $Man$ and $Woman$ are defined to be disjoint in the \textsc{FINANCIAL} dataset, then testing atom $Woman(key)$ as a refinement of the pattern $Q(key)=Client(key), Man(key)$ is useless, and such refinement is pruned due to its unsatisfiability. The important property in the context of the presence of disjointness constraints is also the co-occurrence of the role domain and range specifications. For instance, an atom with a concept describing a variable that is already described by the range of some role already present in a query, may be pruned away if the concept in the role range and the given concept are disjoint. For example, in the \textsc{FINANCIAL} dataset, the range of role $hasCreditCard$ is $CreditCard$. Then, due to the disjointness of concepts $CreditCard$ and $Account$, a refinement $Account(x_1)$ of query $Q(key)=Client(key), hasCreditCard(key, x_1)$ is pruned after the satisfiability test.   
\\
\indent
Furthermore, after the analysis of the results, we conclude that the features of the intensional part of the knowledge base of the tested datasets, that helped to prune patterns after the \emph{s-freeness} test were: the organization of concepts and roles in taxonomies, the specification of domain and ranges of roles, the specification of role properties, such as role inverse, and concept definitions.
\\
\indent
In the \textsc{FINANCIAL} dataset the hierarchy of roles is flat, the hierarchy of concepts maximum 4 levels deep. In the r\textsc{SWRC} dataset the hierarchy of roles is also flat, and the hierarchy of concepts is maximum 5 levels deep. In the r\textsc{LUBM} dataset the hierarchy of roles is almost flat, with two exceptions, 2 and 3 levels deep, the hierarchy of concepts is maximum 5 levels deep.
\\
\indent
In the \textsc{FINANCIAL} dataset all roles have domain and ranges explicitly specified. There are also some axioms defining inverse roles. In the r\textsc{SWRC} dataset domains and ranges are nearly not specified explicitly (except with one exception). There are, however, restrictions imposed on ranges of some roles while used with particular concepts in a role domain, for example when concept $AcademicStaff$ is used as a domain of the role $headOfGroup$ the range of the role can only be $ResearchGroup$: $AcademicStaff \sqsubseteq \forall headOfGroup$.$ResearchGroup$.
There are many axioms specifying inverse roles in r\textsc{SWRC}.
In r\textsc{LUBM} dataset domains and ranges of some roles are explicitly specified, there are also some inverse role specifications. 
\\
\indent
Concept definitions are only present in the r\textsc{LUBM} dataset. The example concept definition of concept $Student$ as \emph{a person taking some course} is: $Student \equiv Person \sqcap \exists takesCourse$.$Course$.
Hence, the atom $Student(key)$ as the refinement of query $Q(key)=Person(key), takesCourse(x_1), Course(x_1)$ is pruned after the s-freeness test.
\\
\indent
In the context of the query \emph{equivalency} test (when performed after the s-freeness test), the important features of the intensional part of the tested datasets are: specification of role properties, such as role inverse, and concept definitions. 
\\
\indent
Let us consider for example the following queries, $Q_1$ and $Q_2$, tested for semantic uniqueness w.r.t. the \textsc{FINANCIAL} dataset:
\\\\
$Q_1(key)=Client(key), isOwnerOf(key, x_1), hasPermanentOrder(x_1, x_2)$\\
$Q_2(key)=Client(key), hasOwner(x_1, key), hasPermanentOrder(x_1, x_2)$\\
\\
Since the role $hasOwner$ is an inverse of the role $isOwnerOf$, both patterns have exactly the same meaning, and one of them is semantically redundant.
\\
\indent
Concluding, from the analysis of the results and of the tested dataset features, it follows that the presence or lack of the disjointness constraints in a dataset, is a crucial feature. It is also desirable that disjointness constraints co-occur together with the specification of role domains and ranges. Disjointness constraints allow to prune unsatisfiable patterns before any query answering procedure execution, either on $(KB, P)$ or on a copy of $(KB, P)$.
\\
\indent
Moreover, the question posed in Section \ref{sec:ex_nosemresults}, why the r\textsc{SWRC} dataset is especially hard for our approach, has been indirectly answered by the analysis presented in this section. The r\textsc{SWRC} is not very expressive, as it does not contain disjointness constraints, nor explicit role domain and range specifications, any concept definitions, and its role hierarchy is flat. Hence, the gain that may be achieved by using intensional background knowledge for this dataset cannot be large. 
\subsubsection{Using taxonomies of concepts and roles in building pattern refinements}\label{sec:tax}
Table \ref{tab:results_tax} presents the results of the experiment on the use of concept and role taxonomies to build pattern refinements in line with hierarchical information from a $KB$ ($SEM+TAX$) in addition to the setting $SEM$. 
For all the datasets speedup has been achieved, growing with the increasing maximum length of patterns. 

\begin{table}[t]
\caption{Results of the experiment. \textsc{FINANCIAL}: $minsup$=0.2, $\hat{C}$=$Client$, r\textsc{SWRC}: $minsup$=0.3, $\hat{C}$=$Person$, r\textsc{LUBM}: $minsup$=0.3, $\hat{C}$=$Person$}\label{tab:results_tax}
\footnotesize
\begin{center}
\begin{tabular}{r | r | r | r | r  }
\hline
\multicolumn{1}{c|} {Dataset} & {Max} & \multicolumn{2}{c|}{runtime[s]} & \multicolumn{1}{c} {speedup}\\
& Length & \multicolumn{1}{r}{SEM} &	SEM+TAX &	SEM/SEM+TAX\\	
\hline
\textsc{FINANCIAL} & 12 &	30821.6 & 23146.4	& \textbf{1.33} \\
r\textsc{SWRC} & 4 &   2533.5 &	 1594.4 & \textbf{1.59}\\
r\textsc{LUBM} & 4 &   3486.7 &	2232.4 & \textbf{1.56}\\
\hline
\end{tabular}
\end{center}
\end{table}

\section{Related work}\label{sec:related}
We start the discussion in this section, from the features of knowledge representation languages admitted by the (onto-)relational frequent pattern mining approaches. Further we discuss how the related approaches exploit the semantics of their admitted representation languages. 
\\
\indent
The relational frequent pattern miners, \textsc{WARMR}, \textsc{FARMER}, and c-armr\index{c-armr}, are designed to operate on knowledge bases represented as logic programs (generally in a Datalog variant).  
Thus, by definition, they are not able to use the knowledge that has the form of description logic axioms that could not be rewritten into Datalog.  
Consider the knowledge base from Example \ref{ex:KB_statement}. In Datalog it is impossible to assert that all accounts must have an owner without explicitly specifying who the owner is. Let us show how this may affect the properties of patterns generated by a method. Consider the following pattern (query):
$$Q(x) =? - Account(x),Property(x)$$
The second atom of $Q$ may be considered as redundant, as from the knowledge base we already know that every account is a property, and we know also that if something has an owner, than it is a property. Moreover, it is stated that every account has an owner. For example, $account2$ is a property, even if there is nowhere written so, and the owner of $account2$ is nowhere specified. In Datalog such deduction, in order to catch the redundancy, is impossible.
\\
\indent
DL-safe rules formalism allows modelling the rules with disjunctions of atoms in their heads. 
Hence, despite of admitting an additional component of the knowledge base (in description logic), we also extend the language of logic programs used by other methods from Datalog to disjunctive Datalog. 
Summarizing, \textsc{WARMR}, \textsc{FARMER} and c-armr\index{c-armr} operate only on a part of the one of the two components assumed in our approach, namely only on the Datalog part of disjunctive Datalog.
\\
\indent
\textsc{SPADA} (in further versions named $\mathcal{AL}$-QuIn) has been the only approach so far to frequent pattern discovery in combined knowledge bases, more specifically the knowledge bases expressed in $\mathcal{AL}$-log \cite{AL_log}. $\mathcal{AL}$-log is the combination of Datalog with $\mathcal{ALC}$ description logic, and hence our approach supports more expressive language $\mathcal{SHIF}$ in the description logic component. The early version of \textsc{SPADA}/$\mathcal{AL}$-QuIn admitted only $\mathcal{ALC}$ atomic concepts as the structural knowledge (i.e., taxonomies) whereas roles and complex concepts have been disregarded. $\mathcal{AL}$-QuIn does not have this restriction. Rules in \textsc{SPADA}/$\mathcal{AL}$-QuIn are represented as constrained Datalog clauses. In these clauses, only DL concepts can be used (as constraints in the body). While in DL-safe rules using both concepts and roles in DL-atoms is allowed and DL-atoms can be used in rule heads as well. 
DL-safe rules are applicable only to explicitly named objects. The fact that atoms with concept predicates can occur only as constraints in the body of $\mathcal{AL}$-log rules has the similar effect. 
\\
\indent
The actual representation considered in the \textsc{SPADA}/$\mathcal{AL}$-QuIn is a special kind of Datalog (where description logic concepts serve as constraints in the Datalog clauses) \cite{Lisi_OWLED}. In the core of $\mathcal{AL}$-QuIn, description logic axioms are compiled into Datalog ones and appended to Datalog component (for the details see: \cite{Lisi_OWLED}). In turn, the DL-safe rules extend expressive description logics with disjunctive clauses. In the core of the reasoning mechanism proposed for DL-safe rules there is an inverse direction: description logic knowledge base is translated into a disjunctive Datalog program and rules are appended to the result of this translation. Note, that the first approach, consisting on computing consequences of the DL component first, and then applying the rules to these consequences is incorrect in general. 
Consider the following knowledge base $(KB, P)$:\\
\\
\ \ $Human(x) \leftarrow Man(x), \mathcal{O}(x)$ \\
\ \ $Human(x) \leftarrow Woman(x), \mathcal{O}(x)$ \\
\ \ $(Man \sqcup Woman)(Pat)$. \\
\\
The assertion of individual $Pat$ to concept $(Man \sqcup Woman)$ means that $Pat$ is a $Man$ or a $Woman$, but it is not known whether Pat is a man or a woman. Either of the rules from $(KB, P)$ derives that $Pat$ is a $Human$, hence $(KB, P) \models Human(Pat)$. It could not be derived by applying these rules to the consequences of $KB$, since $KB \not{\models} Man(Pat)$ and $KB \not{\models} Woman(Pat)$.  
Thus, in order to perform the translation from the description logic to Datalog, $\mathcal{AL}$-QuIn has to assure that \emph{all} the concepts are named, which is not a restriction in our approach. 
\\
\indent
Eventually, the language used in \textsc{SPADA}/$\mathcal{AL}$-QuIn corresponds to Datalog, while that used in our approach to the, more expressive, disjunctive Datalog. 
\\
\indent
The relational data mining methods, \textsc{WARMR} and \textsc{FARMER}, use  $\theta$-\emph{subsumption} as the generality measure. The $\theta$-subsumption is a syntactic generality relation and as such it is not strong enough to capture semantic redundancies. For the knowledge base from Example \ref{ex:KB_statement}, \textsc{WARMR} may discover queries like the following one:
$$Q(x, y, z)=?-p\_familyAccount(x, y, z), p\_sharedAccount(x, y, z)$$
In the knowledge base, $p\_familyAccount$ is defined as a type of $p\_sharedAccount$. This makes the second atom of $Q$, $p\_sharedAccount(x, y, z)$, semantically redundant. 
The $\theta$-subsumption is to weak to use such taxonomic information.
Using the syntactic generality measure causes redundancy not only in a single pattern, but also in a set of patterns. 
Consider, for example, the following queries that would be both discovered by \textsc{WARMR}:
\begin{center}
$Q_1(x, y, z) = ?-p\_familyAccount(x, y, z)$\\
$Q_2(x, y, z) = ?-p\_familyAccount(x, y, z), p\_sharedAccount(x, y, z)$
\end{center}
Under a semantic generality measure they would be equivalent to each other.
\\
\indent
Both c-armr and \textsc{SPADA} have been conceived to use a semantic generality measure, but \textsc{SPADA} does not fully exploits it in an algorithm for pattern mining to avoid generation of semantically redundant patterns. 
It is not used either to prune patterns semantically redundant, due to redundant literals nor to prune semantically equivalent patterns. That is, similar pattern as described above w.r.t. \textsc{WARMR}, would be generated by \textsc{SPADA}:
$$q(y)\leftarrow p\_woman(y), p\_familyAccount(x, y, z), p\_sharedAccount(x, y, z) \& Client(y)$$
Also, there is no solution in \textsc{SPADA} algorithm to check the redundancy using the knowledge linking Datalog and description logic component (like the rule defining $p\_familyAccount$). The following clause may be generated:   
$$q(x)\leftarrow p\_familyAccount(x, y, z), p\_woman(y) \ \&\ Account(x), Client(y)$$
while from the $(KB, P)$ already follows the constraint $Client$ on variable $y$. Since the second argument of $p\_familyAccount$ describes an owner, and being an owner implies being a client, the atom $Client(y)$ is redundant w.r.t. $(KB, P)$.   
\\
\indent
It is also interesting to note here, that 
in c-armr not only semantically free (s-free), but also semantically closed (s-closed) patterns are generated. Generation of semantically closed patterns is based on the assumption that each s-free clause has a unique s-closed clause (s-closure) and several s-free clauses may have the same s-closure. This is valid for Datalog, as s-closures are computed based on the computation of the least Herbrand models. In the context of our approach, this assumption is not valid anymore. A combined knowledge base $(KB, P)$ is translated into a disjunctive Datalog program, where there may be possibly many minimal models. Additionally, if there are transitive roles defined in a knowledge base, in order to be s-closed, a query should contain the transitive closure of atoms with such role. It may cause that a query contains many, possibly not interesting atoms and leads to additional computations. Taking into account the abovementioned issues, we decided to generate only s-free queries. 
\\
\indent
Relational frequent pattern mining algorithms usually generate patterns according to a specification in a \emph{declarative bias}. Declarative bias allows to specify the set of atom templates describing the atoms to be used in patterns. Common solution is to take one atom template after the other to build the refinements of a pattern, in order in which the templates are stored in a declarative bias directives. Such solution is adopted in \textsc{WARMR}, \textsc{FARMER}, c-armr. However, the solution does not make use of a semantic relationships between predicates in atoms, causing redundant computations. Consider the patterns:
\\
\\
\ \ \ \ $Q_1(y) = ?-p\_woman(y), p\_sharedAccount(x, y, z)$\\
\ \ \ \ $Q_2(y) = ?-p\_woman(y), p\_familyAccount(x, y, z)$
\\
\\
Assume that pattern $Q_1$ has been found infrequent. Thus, generating pattern $Q_2$ is useless, since $p\_familyAccount$ is more specific than $p\_sharedAccount$. However, c-armr would generate and test both queries anyway. 
If some taxonomic information was used to systematically generate refinements, the redundant computation could be avoided. 
In \textsc{SPADA}, taxonomic information is used only with regard to the concept hierarchies. That is, patterns are refined by replacing more general concepts by more specific ones in the constraints of the constrained Datalog clause. Any technique using taxonomic information is not reported with regard to the Datalog predicates. It means, the same scheme of refining patterns, described above, is applied in \textsc{SPADA}/$\mathcal{AL}$-QuIn, too. 
\\
\indent
Finally, in Table \ref{tab:features_comparison}, we provide the comparison of the semantic features of the approaches to (onto-)relational frequent pattern mining. In the last row we provide the features of our approach, \textsc{SEMINTEC}.   
\\
\indent
\begin{table}[t]
\caption{Semantic features of (onto-)relational frequent pattern mining methods.}\label{tab:features_comparison}
\scriptsize
\begin{center}
\begin{tabular*}{9.7cm}{|l | l |l | l | l | l | l| l |}
\hline 
& \multicolumn{3}{c|}{Knowledge representation} & \multicolumn{4}{c|}{Method} \\
\hline
  & {\begin{sideways}DL component\end{sideways}}   
 & \rotatebox{90}{\parbox[c]{2cm}{Datalog \ \ \ \ \ \ component}}
 & \rotatebox{90}{\parbox{2cm}{disjunctive rules}}
 & \rotatebox{90}{\parbox{2cm}{semantic \ \ \ \ \ \ \ \  generality     \ \ \ measure}} 
 & \rotatebox{90}{\parbox{2.1cm}{semantically free \ \ \ (non-redundant) patterns}}
 & \rotatebox{90}{\parbox{2cm}{only one pattern from each equivalence class}}
 & \rotatebox{90}{\parbox{2.1cm}{taxonomies\ \ \ \ \ \ \ \ \  directly used in refining patterns}}\\
\hline
\textsc{WARMR} & & x & & & & &\\
\textsc{FARMER} & & x & & & & &\\
c-armr & & x & & x & x & x & \\
\textsc{SPADA}/$\mathcal{AL}$-QuIn\ \ & x & x &  &x & &  & x \\
\textsc{SEMINTEC} & x & x & x & x & x & x & x\\
\hline
\end{tabular*}
\end{center}
\end{table}
With regard to the languages used, all of the approaches are able to operate on a relational component. Only \textsc{SPADA}/$\mathcal{AL}$-QuIn and our proposed  approach, \textsc{SEMINTEC}, are the systems designed to take a description logic component into account. Moreover, \textsc{SPADA}/$\mathcal{AL}$-QuIn is only able to use concepts in patterns, but it is not able to use any roles. Additionally, only the representation used in \mbox{\textsc{SEMINTEC}} allows disjunctions of atoms in rule heads. With regard to the features of the algorithms, only c-armr, \textsc{SPADA}/$\mathcal{AL}$-QuIn and \textsc{SEMINTEC} use the semantic generality measure, with the consequences described earlier in this section. Only \mbox{c-armr} and \mbox{\textsc{SEMINTEC}} apply a technique to prune semantically redundant literals from patterns, by checking s-freeness property. Summarizing, \textsc{SEMINTEC} is the only approach having all the features presented in Table \ref{tab:features_comparison}.

\section{Conclusions and future work}\label{sec:conc}
In this paper we have proposed a new method for frequent pattern discovery from knowledge bases represented in a combination of $\mathcal{SHIF}$ description logic and DL-safe rules.
We have focused on the relation of the semantics of the representation formalism to the task of frequent pattern discovery, as this is a key aspect to the design of (onto-)relational frequent pattern discovery methods.
For the core of our method we have proposed an algorithm that applies techniques that exploit the semantics of the combined knowledge base.
\\
\indent
We have developed a proof-of-concept implementation of this method using the state-of-the-art reasoning techniques. 
We have empirically shown that using the intensional part of the combined knowledge to perform semantic tests on candidate patterns can make data mining faster. This is because the semantic tests help to prune useless patterns before their evaluation, and they help to avoid the futile search of large parts of the pattern space. We have also shown that exploiting the semantics of a knowledge base can improve the quality of the set of patterns produced: the patterns are more compact through the removal of redundant atoms, and more importantly, there are fewer patterns, as only one pattern is produced from each semantic equivalence class.
\\
\indent
The primary motivation for our work is the real-world need of the Semantic Web for data-mining methods.  For example large amounts of biological data are now being represented using descriptions logics and rules, and there is a scientific need to find frequent patterns in this data. Our method is a baseline for future work in this area that may be twofold. Firstly, after careful investigation of a particular needs of prominent application domains, the scope of the method may be extended by considering more expressive languages falling into the framework of DL-safe rules. Secondly, we plan to develop optimization techniques and heuristic algorithms for which the proposed method (complete w.r.t. to the pattern space search) would be a point of reference.
\\\\
\small
\textbf{Acknowledgements.} The authors acknowledge the support of the Polish Ministry of Science and Higher Education (grant number N N516 186437). 
We are grateful to Boris Motik for explanations on DL-safe rules formalism, and for providing KAON2 reasoner, to Prof. Ross D. King for valuable discussions on KR and data mining in biology, and the remarks on earlier versions of the paper, to Jan Ramon for explanations on the concept of the semantic freeness of Datalog patterns. We are also very grateful to the anonymous reviewers for all their comments.
\normalsize

\bibliography{tplp_orai2009}

\end{document}